\documentclass[12pt]{article}

\usepackage{acro}
\usepackage{amsmath}
\usepackage{amssymb}
\usepackage{amsthm}
\usepackage{bm}
\usepackage{bbm}
\usepackage{color}
\usepackage{enumitem}
\usepackage{fullpage}
\usepackage{graphicx}
\usepackage[hidelinks]{hyperref}
\usepackage{lmodern}
\usepackage{natbib}
\usepackage{tikz}
\usepackage{cleveref}
\usepackage[normalem]{ulem}

\hypersetup{
    colorlinks=false
    }

\DeclareMathOperator*{\argmin}{arg\,min}
\DeclareMathOperator*{\rightarrowas}{\stackrel{\textup{a.s.}}{\longrightarrow}}
\DeclareMathOperator*{\rightarrowd}{\stackrel{\textup{d}}{\longrightarrow}}
\DeclareMathOperator*{\rightarrowp}{\stackrel{\textup{p}}{\longrightarrow}}

\newcounter{assumptionsA}
\newlist{assumptionsA}{enumerate}{2}
\setlist[assumptionsA]{
	label={\textup{\textbf{A\arabic*}}},
	ref={\textup{\textbf{A\arabic*}}},
	before=\setcounter{assumptionsAi}{\value{assumptionsA}},
	after=\setcounter{assumptionsA}{\value{assumptionsAi}}
	}

\newcounter{assumptionsB}
\newlist{assumptionsB}{enumerate}{2}
\setlist[assumptionsB]{
	label={\textup{\textbf{B\arabic*}}},
	ref={\textup{\textbf{B\arabic*}}},
	before=\setcounter{assumptionsBi}{\value{assumptionsB}},
	after=\setcounter{assumptionsB}{\value{assumptionsBi}}
	}

\newcounter{assumptionsC}
\newlist{assumptionsC}{enumerate}{2}
\setlist[assumptionsC]{
	label={\textup{\textbf{C\arabic*}}},
	ref={\textup{\textbf{C\arabic*}}},
	before=\setcounter{assumptionsCi}{\value{assumptionsC}},
	after=\setcounter{assumptionsC}{\value{assumptionsCi}}
	}

\newcounter{counter}
\newtheorem{definition}[counter]{Definition}
\newtheorem{theorem}[counter]{Theorem}
\newtheorem{lemma}[counter]{Lemma}

\newtheorem{corollary}[counter]{Corollary}

\newtheorem{proposition}[counter]{Proposition}

\DeclareAcronym{mmd}{short = MMD, long = maximum mean discrepancy}
\DeclareAcronym{qmc}{short = QMC, long = quasi Monte Carlo}
\DeclareAcronym{pdf}{short = PDF, long = probability density function}
\DeclareAcronym{ksd}{short = KSD, long = kernel Stein discrepancy}
\DeclareAcronym{mcmc}{short = MCMC, long = Markov chain Monte Carlo}
\DeclareAcronym{smc}{short = SMC, long = sequential Monte Carlo}

\begin{document}

\title{Minimum Kernel Discrepancy Estimators}
\author{Chris. J. Oates \\
\small Newcastle University, UK} 
\maketitle

\begin{abstract}
For two decades, reproducing kernels and their associated discrepancies have facilitated elegant theoretical analyses in the setting of quasi Monte Carlo.
These same tools are now receiving interest in statistics and related fields, as criteria that can be used to select an appropriate statistical model for a given dataset.
The focus of this article is on \emph{minimum kernel discrepancy estimators}, whose use in statistical applications is reviewed, and a general theoretical framework for establishing their asymptotic properties is presented.
\end{abstract}

\section{Introduction}

The abstract problem tackled in quasi Monte Carlo is to construct an empirical distribution $P_n$ that provides an accurate approximation to a given probability distribution $P$ of interest.
To make this problem well-defined a \emph{discrepancy} is required, to precisely quantify the quality of the approximation $P_n$ to $P$.
The mathematical tractability of the discrepancy is an important consideration, both for algorithmic design and for analysis of the error of the approximations $P_n$ that are produced.
In a seminal paper, \citet{hickernell1998generalized} advocated \emph{kernel discrepancies} in this context, explaining how reproducing kernels can be used to construct explicit discrepancies with a range of favourable mathematical properties.
Subsequent researchers have exploited kernel discrepancies to analyse a variety of modern and powerful quasi Monte Carlo methods \citep[see e.g.][]{dick2010digital}, and more recently also Markov chain Monte Carlo methods \citep[see e.g.][]{gorham2017measuring}.

The abstract problem tackled in statistics is to construct a probabilistic model $P$ for a population quantity of interest, given a finite sample from this population with empirical distribution $P_n$.
Thus, on the face of it, the problems tacked in quasi Monte Carlo and in statistics are diametrically opposed.
Nevertheless, an equally valid approach to the statistical problem is to first select a discrepancy and proceed to select a probabilistic model $P_\theta$ from a parametric set such that the discrepancy between $P_n$ and $P_\theta$ is minimised.
It is well-known that maximum likelihood estimation asymptotically minimises the Kullback--Leibler divergence between $P$ and $P_\theta$ \citep{Akaike1973}, and as such the maximum likelihood estimator can be acutely sensitive to the extent to which the model is misspecified, referring to the scenario in which $P$ and $P_\theta$ do not coincide for any value $\theta$ in the parameter set.
Several more robust (i.e. less sensitive, in the sense just described) alternatives to maximum likelihood have been proposed, such as M-estimation \citep{huber1964robust} and the more general minimum distance estimation \citep{Donoho1988}, which remains an active research field \citep{Basu2011,Pardo2018}.
More interestingly, as far as this article is concerned, \emph{minimum kernel discrepancy estimation} has recently been proposed and studied.
The use of kernel discrepancies in the statistical context can be traced back at least to \citet{Song2008}, before receiving considerable attention in the machine learning community, where minimum kernel discrepancy estimation has been used to train machine learning models \citep{dziugaite2015training,Li2015,Sutherland2017}, to calibrate computer simulations \citep{park2016k2,mitrovic2016dr,dellaporta2022robust}, and to facilitate goodness-of-fit testing \citep{key2021composite}.
It seems, to this author, that there could be a useful dialogue between the quasi Monte Carlo and the statistical communities, centring on whether the theoretical insight into kernel discrepancies gleaned in the former can be used to explain the remarkable success in applications achieved by the latter.
This article serves as an introduction to minimum kernel discrepancy estimation, and along the way some opportunities for dialogue between the communities will be highlighted.

In \Cref{sec: mkdes} both kernel discrepancy and of minimum kernel discrepancy estimators are defined.
The use of minimum kernel discrepancy estimators in statistics is described in \Cref{sec: applications}.
A concise set of theoretical results for minimum kernel discrepancy estimation are presented in \Cref{sec: theory}.
Some open questions and possible research directions are suggested in \Cref{sec: summary}.

\section{Kernels, Discrepancies, and Estimators}
\label{sec: mkdes}

\Cref{subsec: setup} introduces the set-up and notation that will be used.
The role of kernel discrepancies in quasi Monte Carlo is reviewed in \Cref{subsec: kernels in QMC}.
Then, minimum kernel discrepancy estimators are defined in \Cref{subsec: mkde}.

\subsection{Set-Up and Notation}
\label{subsec: setup}

Our setting is a measurable space $\mathcal{X}$.
A \emph{kernel} is a function $k : \mathcal{X} \times \mathcal{X} \rightarrow \mathbb{R}$, that is \emph{measurable}, \emph{symmetric} and \emph{positive semi-definite}, meaning respectively that
\begin{itemize}
\item $k(\cdot,x)$ is measurable for all $x \in \mathcal{X}$
\item $k(x,y) = k(y,x)$ for all $x,y \in \mathcal{X}$
\item $\sum_{i=1}^n \sum_{j=1}^n w_i w_j k(x_i,x_j) \geq 0$ for all $w_1,\dots,w_n \in \mathbb{R}$, $x_1,\dots,x_n \in \mathcal{X}$, $n \in \mathbb{N}$ .
\end{itemize}
The \emph{reproducing kernel Hilbert space} associated to a kernel $k$ is a Hilbert space $\mathcal{H}(k)$ of real-valued functions on $\mathcal{X}$, such that
\begin{itemize}
\item $k(\cdot,x) \in \mathcal{H}(k)$ for all $x \in \mathcal{X}$
\item $\langle h , k(\cdot,x) \rangle_{\mathcal{H}(k)} = h(x)$ for all $x \in \mathcal{X}$, $h \in \mathcal{H}(k)$ .
\end{itemize}
The latter requirement is known as the \emph{reproducing property}.
It can be shown that $\mathcal{H}(k)$ exists, that $\mathcal{H}(k)$ is uniquely determined, and that the elements of $\mathcal{H}(k)$ are measurable functions on $\mathcal{X}$ \citep[see Chapter 4 of][]{steinwart2008support}.
Our interest in reproducing kernel Hilbert spaces derives from their suitability for constructing a discrepancy between probability distributions on $\mathcal{X}$, as explained next.

Let $\mathcal{P}_k(\mathcal{X})$ be the set of probability distributions $P$ on $\mathcal{X}$ for which $\mathrm{P} : \mathcal{H}(k) \rightarrow \mathbb{R}$, $\mathrm{P}(h) = \int h \; \mathrm{d}P$, is a bounded linear functional.
Sufficient conditions for $P \in \mathcal{P}_k(\mathcal{X})$ are discussed in \Cref{app: kernel embed}.
If $P \in \mathcal{P}_k(\mathcal{X})$ then there exists a Riesz representer $\mu_k(P) \in \mathcal{H}(k)$, called the \emph{kernel mean element}, such that $\mathrm{P} (h) = \langle h , \mu_k(P) \rangle_{\mathcal{H}(k)}$ for all $h \in \mathcal{H}(k)$.
The \emph{kernel discrepancy} on $\mathcal{P}_k(\mathcal{X})$ is defined as
\begin{align*}
D_k : \mathcal{P}_k(\mathcal{X}) \times \mathcal{P}_k(\mathcal{X}) & \rightarrow [0,\infty) \\
(P,Q) & \mapsto \| \mu_k(P) - \mu_k(Q) \|_{\mathcal{H}(k)} ,
\end{align*}
being the distance between the Riesz representers $\mu_k(P)$ and $\mu_k(Q)$ in $\mathcal{H}(k)$.
It is immediate that $D_k$ is a pseudo-metric on $\mathcal{P}_k(\mathcal{X})$, meaning that $D_k$ satisfies all the requirements of a metric on $\mathcal{P}_k(\mathcal{X})$ except for the \emph{distinguishability} requirement, since in general we may have $D_k(P,Q) = 0$ even if $P$ and $Q$ are distinct.
The pseudo-metric $D_k$ appears in several fields; in machine learning it is called the \emph{maximum mean discrepancy} \citep{gretton2012kernel,muandet2016kernel}, in probability it is called an \emph{integral probability (pseudo-)metric} \citep{muller1997integral}, in statistics it is the discrepancy associated to the \emph{kernel scoring rule} \citep{gneiting2007strictly,dawid2007geometry}, and in quasi Monte Carlo it is called the \emph{worst case integration error} \citep{dick2010digital}.
The popularity of kernel discrepancy is due in large part to a well-known explicit formula that enables it to be computed:

\begin{proposition} \label{prop: explicit MMD}
For $P,Q \in \mathcal{P}_k(\mathcal{X})$,
\begin{align*}
D_k(P,Q) = \sqrt{ \iint k_Q(x,y) \mathrm{d}P(x) \mathrm{d}P(y) } ,
\end{align*}
where
\begin{align}
k_Q(x,y) := k(x,y) - \int k(x,y) \mathrm{d}Q(x) - \int k(x,y) \mathrm{d}Q(y) + \iint k(x,y) \mathrm{d}Q(x) \mathrm{d}Q(y) . \label{eq: ktheta def}
\end{align}
\end{proposition}
\begin{proof}
The reproducing property, together with the definition of $\mu_k(P)$ as the Riesz representer of $\mathrm{P}$, leads to $\mu_k(P)(y) = \langle \mu_k(P) , k(\cdot,y) \rangle_{\mathcal{H}(k)} = \mathrm{P}(k(\cdot,y)) = \int k(x,y) \mathrm{d}P(x)$.
Thus the kernel mean element is the weak (or \emph{Pettis}) integral $\mu_k(P)(\cdot) = \int k(x,\cdot) \mathrm{d}P(x)$, for all $P \in \mathcal{P}_k(\mathcal{X})$.
In particular, $\langle \mu_k(P) , \mu_k(Q) \rangle_{\mathcal{H}(k)} = \int \mu_k(P)(y) \mathrm{d}Q(y) = \iint k(x,y) \mathrm{d}P(x) \mathrm{d}Q(y)$ for all $P,Q \in \mathcal{P}_k(\mathcal{X})$.
Therefore, recalling that $k$ is symmetric,
\begin{align*}
D_k(P,Q)^2 & = \langle \mu_k(P) , \mu_k(P) \rangle_{\mathcal{H}(k)} - 2 \langle \mu_k(P) , \mu_k(Q) \rangle_{\mathcal{H}(k)} + \langle \mu_k(Q) , \mu_k(Q) \rangle_{\mathcal{H}(k)} \\
& = \iint k(x,y) \mathrm{d}P(x) \mathrm{d}P(x) - 2 \iint k(x,y) \mathrm{d}P(x) \mathrm{d}Q(y) + \iint k(x,y) \mathrm{d}Q(x) \mathrm{d}Q(y) \\
& = \iint k_Q(x,y) \mathrm{d}P(x) \mathrm{d}P(y) ,
\end{align*}
which completes the proof.
\end{proof}

\noindent It can be verified that the map $k_Q : \mathcal{X} \times \mathcal{X} \rightarrow \mathbb{R}$ in \eqref{eq: ktheta def} is a $Q$-dependent kernel, satisfying $\int k_Q(x,y) \mathrm{d}P(x) = 0$ for all $y \in \mathcal{X}$ whenever $Q = P$.
This presentation, which assigns different roles to $P$ and $Q$, is useful in our setting, where the first argument $P$ will be fixed and the second argument $Q$ will be varied.

\subsection{Kernel Discrepancies in Quasi Monte Carlo}
\label{subsec: kernels in QMC}

Classical analysis of quasi Monte Carlo focused on the design of point sets or sequences for which small values of a \emph{figure of merit}, quantifying in a sense the uniformity of the points, is achieved.
For simplicity the domain $\mathcal{X}$ is typically taken to be $[0,1]^d$.
Depending on the nature of the algorithm being considered, and the assumed regularity of the integrands involved, a different figure of merit was typically used.
For example, the van der Corput--Halton and Sobol' sequences were analysed using the \emph{star discrepancy} figure of merit \citep{hlawka1961funktionen}, while for lattice rules a figure of merit called $P_\alpha$ was used \citep{sloan1987lattice}.
A turning point came in \citet{hickernell1998generalized}, who observed that each of the above figures of merit can be recovered as a specific instance of a kernel discrepancy $D_k$ (i.e. corresponding to a particular choice of kernel $k$).
Further, elementary properties of reproducing kernels give rise to a natural cubature error bound, involving simply the product of $D_k$ and the norm of the integrand in $\mathcal{H}(k)$.
Since then, kernel discrepancies have been adopted for the analysis of a variety of modern and powerful quasi Monte Carlo methods.
Indeed, many of the regularities that are exploited by modern quasi-Monte Carlo algorithms can be encoded directly at the level of the kernel.

Two examples in particular clearly illustrate the elegance of the kernel framework.
First, the \emph{dominating mixed smoothness} assumption, concerning the existence of mixed partial derivatives of order up to $s$ in each coordinate\footnote{i.e. additional levels of differentiability in directions that are axis-aligned}, corresponds to a tensor product kernel
\begin{align*}
k(x,y) = \prod_{i=1}^d k_{\text{1D}}(x_i,y_i)
\end{align*}
where $k_{\text{1D}}$ is a kernel reproducing $s$-times differentiable functions on $[0,1]$ \citep[see e.g.][Section 14.6]{dick2010digital}.
Such assumptions have been used to prove arbitrarily fast convergence rates for quasi Monte Carlo based on \emph{higher-order digital nets}, as the order of differentiability $s$ is increased \citep[see e.g.][Theorem 15.21]{dick2010digital}.
Tensor product kernels have the further advantage that computation of the kernel mean element $\mu_k(P)$ reduces to computation of $d$ univariate integrals whenever the components of a $P$-distributed random variable are independent, which is the case for the canonical choice of the uniform distribution on $[0,1]^d$.
Second, the \emph{effective low dimension} assumption, which states that most of the variation in the output of a function occurs when only a small number of the inputs are varied, can be expressed using a weighted kernel
\begin{align}
k(x,y) = \sum_{\mathfrak{u} \subseteq \{1,\dots,d\}} \gamma_{\mathfrak{u}} k_{\mathfrak{u}}(x_{\mathfrak{u}},y_{\mathfrak{u}}) , \label{eq: weighted kernel}
\end{align}
where $x_{\mathfrak{u}} = (x_i : i \in \mathfrak{u})$ and $k_{\mathfrak{u}}$ is a kernel on $[0,1]^{|\mathfrak{u}|}$. 
The weights $\gamma_{\mathfrak{u}} \geq 0$ determine the extent to which variation in the coordinates indexed by $\mathfrak{u}$ is permitted.
The additive form of \eqref{eq: weighted kernel} enables the associated kernel discrepancy to be computed; indeed, $D_k^2 = \sum_{\mathfrak{u}} \gamma_{\mathfrak{u}} D_{k_{\mathfrak{u}}}^2$ follows from \Cref{prop: explicit MMD}.
The specification of suitable weights for the integration task at hand is an important factor in ensuring low cubature error \citep{kuo2003component}.
In an extreme case the \emph{effective dimension}, defined as $\sum_{\mathfrak{u}} \gamma_{\mathfrak{u}}$, may remain finite as $d \rightarrow \infty$ and one can obtain error bounds for quasi Monte Carlo that are dimension-independent \citep{sloan1998quasi,dick2013high}.

The author, in writing this article, is interested in whether the considerable theoretical and practical expertise in kernel discrepancies that has developed in the quasi Monte Carlo community can be brought to bear also on theoretical and practical aspects of minimum kernel discrepancy estimators, which are introduced next.

\subsection{Minimum Kernel Discrepancy Estimators}
\label{subsec: mkde}

Consider data that are assumed to arise as a sequence of independent and identically distributed random variables $(x_n)_{n \in \mathbb{N}}$, with $x_n \sim P$ and $P \in \mathcal{P}_k(\mathcal{X})$.
Let $P_n = \frac{1}{n} \sum_{i=1}^n \delta_{x_i}$ denote the empirical distribution of the first $n$ terms in the dataset.
A \emph{statistical model} is a collection $\{P_\theta\}_{\theta \in \Theta} \subset \mathcal{P}_k(\mathcal{X})$, where the \emph{parameter} $\theta$ takes values in an index set $\Theta$.
\emph{Minimum kernel discrepancy estimators}, if they exist, are defined as
\begin{align*}
\theta_n \in \argmin_{\theta \in \Theta} D_k(P_n , P_\theta) ,
\end{align*}
as illustrated in \Cref{fig:  cartoon}.
The properties of such estimators depend on the choice of kernel $k$ and the statistical model $\{P_\theta\}_{\theta \in \Theta}$, but are independent of how the statistical model is parametrised (in $\theta$).
Under regularity conditions, established in \Cref{sec: theory}, the estimators $\theta_n$ converge (in an appropriate sense) to a population minimiser
\begin{align*}
\theta_\star & \in \argmin_{\theta \in \Theta} D_k(P,P_\theta) ,
\end{align*}
if these exist.
Since in all cases we take $Q = P_\theta$ in \eqref{eq: ktheta def}, we henceforth adopt the shorthand $k_\theta$ for $k_Q$.
Through selection of the kernel $k$, one can trade statistical efficiency with robustness to model misspecification, which is not possible within the classical framework of maximum likelihood; see \Cref{subsec: related}.
However, important aspects of kernel selection remains unsolved, for example it is unclear how to retain statistical efficiency when $\mathcal{X}$ is high-dimensional.
Such open challenges are highlighted in \Cref{sec: summary}.

\begin{figure}[t!]
\centering

\begin{tikzpicture}
\node[inner sep=0pt] (H) at (0,0)      {\includegraphics[angle=90,width = 0.5\textwidth,clip,trim = 2cm 2cm 2cm 0cm]{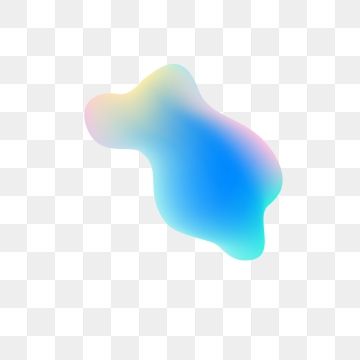}};

\node[] (model) at (1,1) {$\{\mu_k(P_\theta)\}_{\theta \in \Theta}$};

\node[draw,circle,inner sep=2pt,fill] (muPn) at (-3,-1.6) {};
\node[] (muPnlab) at (-3.3,-1.1) {$\mu_k(P_n)$};
\node[draw,circle,inner sep=2pt,fill,label=70:{$\mu_k(P_{\theta_n})$}] (muPtn) at (-2,-1) {};
\node[draw,circle,inner sep=2pt,fill,label=270:{$\mu_k(P)$}] (muP) at (-3.2,-2.5) {};
\node[draw,circle,inner sep=2pt,fill,label=40:{$\mu_k(P_{\theta_\star})$}] (muPt) at (-1.6,-2.2) {};

\draw[|-|,thick] (muPn) -- (muPtn);
\draw[|-|,thick] (muP) -- (muPt);

\node[inner sep=0pt] (d3) at (-7,2.5) {\includegraphics[width = 0.12\textwidth]{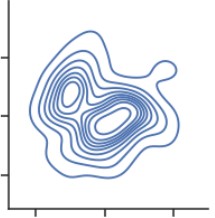}};
\node[inner sep=0pt] (d4) at (-7,0) {\includegraphics[width = 0.12\textwidth]{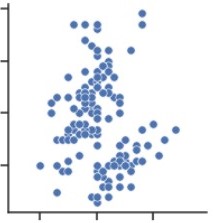}};
\node[inner sep=0pt] (d1) at (-7,-2.5) {\includegraphics[width = 0.12\textwidth]{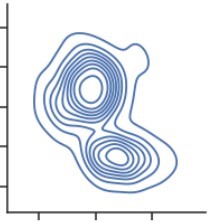}};

\node[inner sep=0pt] at (-8.5,2.5) {$P_{\theta_n}$};
\node[inner sep=0pt] at (-8.5,0) {$P_n$};
\node[inner sep=0pt] at (-8.5,-2.5) {$P$};

\draw[dotted] (d3) to[out=0,in=140] (muPtn);
\draw[dotted] (d4) to[out=0,in=220] (muPn);
\draw[dotted] (d1) to[out=0,in=180] (muP);

\node[inner sep=0pt] at (3.6,-3) {$\mathcal{H}(k)$};

\end{tikzpicture}
\caption{\textit{Minimum kernel discrepancy estimation.} 
The distance, as measured in $\mathcal{H}(k)$, between the kernel mean element $\mu_k(P_\theta)$ associated to the statistical model $P_\theta$ and the kernel mean element $\mu_k(P_n)$ associated to the dataset is minimised. 
Under conditions established in \Cref{sec: theory}, a minimum kernel discrepancy estimator $\theta_n$ converges to a value $\theta_\star$ that minimises the distance to the kernel mean element $\mu_k(P)$ associated to the true data-generating distribution $P$.
 }
 \label{fig:  cartoon}
\end{figure}
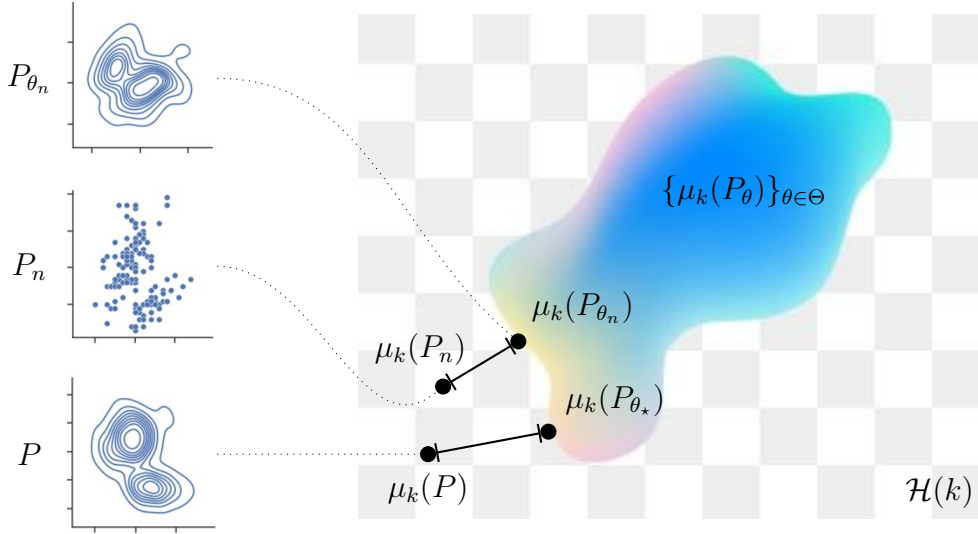

\section{Applications}
\label{sec: applications}

This section presents three distinct applications for minimum kernel discrepancy estimation: the generalised method of moments (\Cref{subsec: summary matching}), generative adversarial networks (\Cref{subsec: minimum MMD}), and energy-based models (\Cref{subsec: minimum KSD}).

\subsection{Generalised Method of Moments}
\label{subsec: summary matching}

Let $\phi : \mathcal{X} \rightarrow \mathbb{R}^p$ have components that are integrable with respect to each $\{P_\theta\}_{\theta \in \Theta}$, and let $\phi(P) := \int \phi \mathrm{d}P_\theta$ denote a vector of \emph{summary statistics}, meaning that this vector summarises some of the salient aspects of $P_\theta$.
The \emph{generalised method of moments} estimator, if it exists, is defined as
\begin{align*}
\theta_n \in \argmin_{\theta \in \Theta} \| \phi(P_n) - \phi(P_\theta) \| ,
\end{align*}
being the parameter vector of a statistical model whose summary statistics agree most closely with the empirical summary statistics $\phi(P_n)$ calculated using the dataset \citep{hansen1982large}.
This procedure can be cast as minimum kernel discrepancy estimation using a kernel $k(x,y) = \langle \phi(x) , \phi(y) \rangle$.
For this kernel, $\int k(x,y) \mathrm{d}P_\theta(y) = \langle \phi(x) , \phi(P_\theta) \rangle$ and $\iint k(x,y) \mathrm{d}P_\theta(x) \mathrm{d}P_\theta(y) = \langle \phi(P_\theta) , \phi(P_\theta) \rangle$.
Then $k_\theta(x,y) = \langle \phi(x) , \phi(y) \rangle - \langle \phi(x) , \phi(P_\theta) \rangle - \langle \phi(P_\theta) , \phi(y) \rangle + \langle \phi(P_\theta) , \phi(P_\theta) \rangle$, so that $D_k(P_n,P_\theta)^2 = \langle \phi(P_n) , \phi(P_n) \rangle - \langle \phi(P_n) , \phi(P_\theta) \rangle - \langle \phi(P_\theta) , \phi(P_n) \rangle + \langle \phi(P_\theta) , \phi(P_\theta) \rangle = \|\phi(P_n) - \phi(P_\theta)\|^2$.
In contrast to minimum kernel discrepancy estimation in general, the generalised method of moments is mathematically transparent in the special case where we assume the statistical model $P_\theta$ can be parametrised such that $\theta = \phi(P_\theta)$ and $\Theta = \mathbb{R}^p$.
Indeed, in this case $D_k(P_n,P_\theta) = \|\theta_n - \theta\|$, so the generalised method of moments estimator is 
$$
\theta_n = \frac{1}{n} \sum_{i=1}^n \phi(x_i), 
$$
which converges to $\theta_\star = \phi(P)$ and satisfies the central limit $\sqrt{n}(\theta_n - \theta_\star) \rightarrowd N(0, \mathbb{C}_{X \sim P}[\phi(X)] )$ whenever the associated covariance matrix is well-defined.
This will provide a convenient opportunity to verify, in a concrete setting, the general theoretical results presented in \Cref{subsec: asym norm}.

Though originating as a statistically efficient and computationally favourable alternative to maximum likelihood, it has more recently been observed that, through careful selection of the summary statistics $\phi$, one can usefully trade-off the efficiency and the robustness of the estimator.
The related areas of \emph{approximate Bayesian computation} \citep{beaumont2019approximate} and \emph{Bayesian synthetic likelihood} \citep[e.g.][]{frazier2021robust} contain extensive practical guidance on how summary statistics can be selected.
Indeed, the estimator $\theta_n$ can be recognised as a \emph{maximum a posteriori} estimator in an approxmate Bayesian computation framework when a flat \emph{a priori} distribution is employed.

\subsection{Generative Adversarial Networks}
\label{subsec: minimum MMD}

A \emph{generative model} consists of a measurable space $\mathcal{Y}$, equipped with a probability measure $\mathbb{P}$, and statistical model of the form
\begin{align}
P_\theta = G_\#^\theta \; \mathbb{P} \label{eq: gen model}
\end{align}
where $G^\theta : \mathcal{Y} \rightarrow \mathcal{X}$ is measurable for each $\theta \in \Theta$.
Here $G_\#^\theta \; \mathbb{P}$ denotes the \emph{pushforward} of $\mathbb{P}$ through $G^\theta$, meaning that samples $X \sim P_\theta$ are generated by first sampling $Y \sim \mathbb{P}$ and then setting $X = G^\theta(Y)$.
In the specific case where $G^\theta$ is a (deep) neural network, we call $P_\theta$ a (deep) \emph{neural network generative model}.
It is clear that many rich and interesting statistical models can be constructed in this manner, using arbitrarily large neural networks with a corresponding parameter vector $\theta$ that is arbitrarily high-dimensional.
However, conventional statistical wisdom suggests that such models may not be useful in practice, due to the prohibitively large number of parameters that will need to be estimated.
Maximum likelihood estimation has shown to struggle when real data are used since, despite its flexibility, the statistical model will necessarily be misspecified \citep{theis2016note}.
Alternative strategies have nevertheless been developed to estimate the parameters $\theta$ of such models and -- remarkably -- these have demonstrated a high level of success provided one has access to a sufficiently rich training dataset.
Here we focus on a line of research called \emph{generative adversarial networks} (GANs) \citep{goodfellow2020generative}, and specifically the so-called \emph{integral probability metric} GANs.
Let $P_n$ denote the empirical distribution of the training dataset.
Then a GAN estimator for $\theta$, if one exists, is defined as
\begin{align*}
\theta_n \in \argmin_{\theta \in \Theta} \; \max_{f \in \mathcal{F}} \; \left| \int f \mathrm{d}P_n - \int f \mathrm{d} P_\theta \right| ,
\end{align*}
where $\mathcal{F}$ is a set of test functions such that $\mathcal{F} \subset L^1(P_\theta)$ for each $\theta \in \Theta$.
The choice of $\mathcal{F}$ determines the properties of the GAN estimator.
The GAN nomenclature originates from the case where the elements of $\mathcal{F}$ are also neural networks, so that the \emph{witness function} or \emph{critic} $f$ for which the maximum is realised can be interpreted as an \emph{adversarial} neural network.
Alternative choices of $\mathcal{F}$ lead to so-called  \emph{generative moment matching networks} \citep{dziugaite2015training,Li2015}, \emph{MMD} GANs \citep{li2017mmd}, \emph{Wasserstein} GANs \citep{arjovsky2017wasserstein}, \emph{Mc} GANs \citep{mroueh2017mcgan}, \emph{Fisher} GANs \citep{mroueh2017fisher}, and \emph{Sobolev} GANs \citep{mroueh2018sobolev}.
Here we focus on generative moment matching networks, which take $\mathcal{F}$ to be the unit ball in a reproducing kernel Hilbert space $\mathcal{H}(k)$ of real-valued functions on $\mathcal{X}$; see \Cref{fig: gan}.
It is a standard fact that this procedure can be cast as minimum kernel discrepancy estimation:

\begin{proposition}
\label{prop: GAN is MMD}
The GAN estimator $\theta_n$ is a minimum kernel discrepancy estimator when $\mathcal{F}$ is the set of $f \in \mathcal{H}(k)$ such that $\|f\|_{\mathcal{H}(k)} \leq 1$.
\end{proposition}
\begin{proof}
For $f \in \mathcal{F}$, using Cauchy--Schwarz,
\begin{align*}
\left| \int f \mathrm{d}P_n - \int f \mathrm{d} P_\theta \right| & = \left| \langle f , \mu_k(P_n) - \mu_k(P_\theta) \rangle \right|  \leq \| f \|_{\mathcal{H}(k)} \|\mu_k(P_n) - \mu_k(P_\theta) \|_{\mathcal{H}(k)} 
\end{align*}
which shows that
\begin{align*}
0 \leq \max_{f \in \mathcal{F}} \; \int f \mathrm{d}P_n - \int f \mathrm{d} P_\theta \leq D_k(P_n,P_\theta) .
\end{align*}
If $D_k(P_n,P_\theta) = 0$ then we must have equality, so suppose $D_k(P_n,P_\theta) \neq 0$.
Then one can verify that $f = [\mu_k(P_n) - \mu_k(P_\theta)] / D_k(P_n,P_\theta)$ is the witness function for which equality is realised.
\end{proof}

This class of GANs is appealing in the sense that the optimisation over $\mathcal{F}$ can be explicitly solved, and the solution is the kernel discrepancy $D_k(P_n,P_\theta)$.
Although the kernel mean element $\mu_k(P_\theta)$ is not closed-form in general, it can be consistently approximated using an empirical approximation to $P_\theta$; this is discussed further in \Cref{sec: summary}.
The performance of these GANs depends on the choice of kernel $k$, and methods have been proposed to flexibly \emph{learn} a suitable $k$ based on the training dataset \citep{li2017mmd}.
These methods have demonstrated remarkable performance on machine learning benchmarks, out-performing other forms of GAN \citep{binkowski2018demystifying}.

\begin{figure}[t!]
\centering

\begin{tikzpicture}
\node[inner sep=0pt] (cont) at (0,0)      {\includegraphics[angle=90,width = 0.3\textwidth]{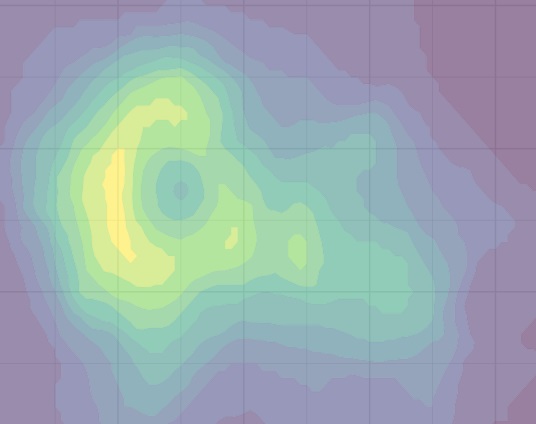}};

\node[] (model) at (-2,2.8) {$\Theta$};

\node[inner sep=0pt] (bad) at (7,1.5) {\includegraphics[width = 0.2\textwidth]{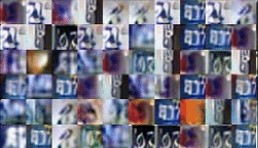}};
\node[inner sep=0pt] (good) at (7,-1.5) {\includegraphics[width = 0.2\textwidth]{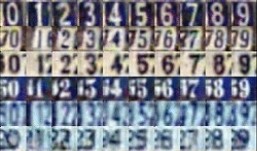}};

\node[draw,circle,inner sep=2pt,fill,label=225:{$\theta$}] (t) at (0.9,1.5) {};
\node[draw,circle,inner sep=2pt,fill] (tn) at (0,-1.8) {};

\draw[dotted] (tn) to[out=15,in=180] (good);
\draw[dotted] (t) to[out=45,in=180] (bad);

\draw[rotate around={10:(0,-1.8)},thick] (0,-1.8) ellipse (0.7cm and 0.2cm);
\draw[dashed,rotate around={10:(0,-1.8)}] (0,-1.8) ellipse (1.4cm and 0.4cm);

\node[circle,fill=white,inner sep=1pt,opacity=.4,text opacity=1] (tnlab) at (0.5,-2.2) {$\theta_n$};

\node[] (C) at (-0.8,-1.3) {$C_n^{0.95}$};

\node[] (Pt) at (4.8,2.2) {$P_\theta$};
\node[] (Ptn) at (4.8,-0.8) {$P_{\theta_n}$};

\end{tikzpicture}
\caption{\textit{Minimum kernel discrepancy estimation for generative models.} The parameters $\theta$ of sophisticated generative models may be estimated by minimising the kernel discrepancy between the generative model $P_\theta$ and the empirical distribution of the dataset.
Confidence sets $C_n^\gamma$ for $\theta_\star$ can be constructed to ensure the notional coverage of $100\gamma\%$ is asymptotically correct; see \Cref{sec: theory}. }
\label{fig: gan}
\end{figure}
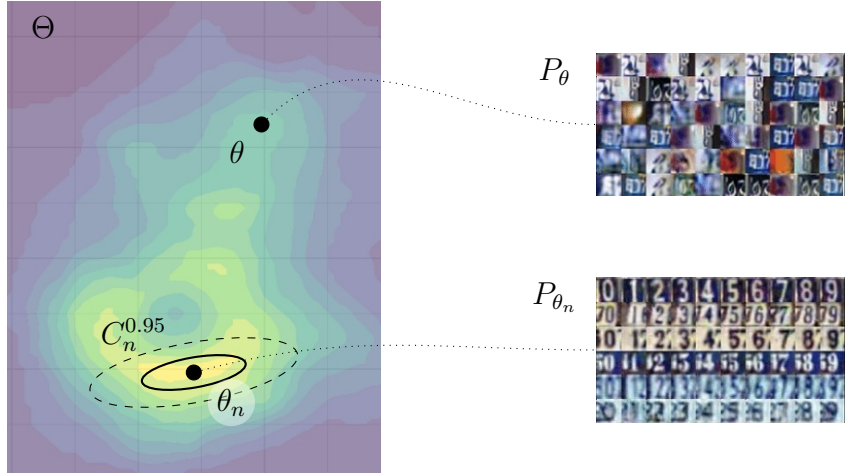

\subsection{Energy-Based Models}
\label{subsec: minimum KSD}

An \emph{energy-based model} on $\mathcal{X} = \mathbb{R}^d$ is a statistical model $P_\theta$ with \ac{pdf} 
\begin{align*}
p_\theta(x) = \frac{\exp(- E_\theta(x))}{Z(\theta)}
\end{align*}
where $E_\theta : \mathcal{X} \rightarrow \mathbb{R}$ is a flexibly parametrised \emph{potential} such that $x \mapsto \exp(-E_\theta(x))$ is integrable and $Z(\theta) := \int \exp(-E_\theta(x)) \mathrm{d}x > 0$.
Energy-based models arise both in statistics, where they are referred to as an instance of \emph{intractable likelihood} \citep{lyne2015russian}, and in machine learning, where they provide an alternative to GANs that can be more convenient when an explicit \ac{pdf} is required \citep{lecun2006tutorial}.
Similarly to GANs, maximum likelihood estimation has been shown to struggle in this setting, where, despite its flexibility, the statistical model will necessarily be misspecified.
However, there is an additional barrier to application of minimum kernel discrepancy estimation to energy-based models; the challenge of approximating the kernel mean element $\mu_k(P_\theta)$.
Indeed, due to the intractable normalisation constant $Z(\theta)$, sampling from an energy-based model is not straight-forward and an algorithm such as \ac{mcmc} is required to approximately sample from $P_\theta$ for each value of $\theta \in \Theta$.
This can entail a prohibitive computational cost when attempting optimisation over $\Theta$ \citep{song2021train}.

To address this computational challenge, we can employ a class of kernels $k$ introduced in \citet{oates2017control}.
These kernels $k$ solve the computational problem by ensuring that $\mu_k(P_\theta)$ is equal to 0 in $\mathcal{H}(k)$, and thus $\mu_k(P_\theta)$ can be trivially computed.
To achieve this, the kernel $k$ must be allowed to depend on $\theta$, and in what follows we write $k(x,y ; \theta)$ to make this dependence explicit.
Before describing the kernels with this property, we first introduce some notation and assumptions.
Let $\nabla f = (\partial_{x_1}f,\dots,\partial_{x_d}f)^\top$ denote the gradient of a differentiable function $f : \mathbb{R}^d \rightarrow \mathbb{R}$ and let $\nabla \cdot f = \partial_{x_1}f_1 + \dots + \partial_{x_d}f_d$ denote the divergence of a differentiable function $f : \mathbb{R}^d \rightarrow \mathbb{R}^d$.
The presentation in the remainder of this section follows \citet{barp2022targeted}:

\begin{proposition} \label{lem: Stein op int by parts}
Let $g : \mathbb{R}^d \rightarrow \mathbb{R}^d$ satisfy $g, \mathcal{A}_\theta g \in L^1(P_\theta)$ where
$$
(\mathcal{A}_\theta g)(x) := (\nabla \cdot g)(x) - \langle g(x) , (\nabla E_\theta)(x) \rangle .
$$
Then $\int \mathcal{A}_\theta g \; \mathrm{d}P_\theta = 0$.
\end{proposition}
\begin{proof}
Let $\varphi_n : \mathbb{R}^d \rightarrow \mathbb{R}$ be a compactly supported function with $\varphi_n(x) = 1$ for $\|x\| \leq n$ and $\|\nabla \varphi_n\| < n^{-1}$, for each $n \in \mathbb{N}$.
From the divergence theorem,
\begin{align*}
0 = \int (\nabla \cdot (\varphi_n p_\theta g))(x) \; \mathrm{d}x 
= \int \langle \nabla \varphi_n(x) , (p_\theta g)(x) \rangle \; \mathrm{d}x + \int \varphi_n(x) (\nabla \cdot (p_\theta g))(x) \; \mathrm{d}x .
\end{align*}
Since $\varphi_n \rightarrow 1$ pointwise and $\nabla \cdot (p_\theta g) \in L^1(\mathbb{R}^d)$, from the dominated convergence theorem 
\begin{align*}
\int \varphi_n(x) (\nabla \cdot (p_\theta g))(x) \; \mathrm{d}x \rightarrow \int (\nabla \cdot (p_\theta g))(x) \; \mathrm{d}x .
\end{align*}
On the other hand, since $p_\theta g \in L^1(\mathbb{R}^d)$,
\begin{align*}
\left| \int \langle \nabla \varphi_n(x) , (p_\theta g)(x) \rangle \; \mathrm{d}x \right| \leq \|\nabla \varphi_n\| \int \|(p_\theta g)(x)\| \; \mathrm{d}x \rightarrow 0 .
\end{align*}
Thus we have shown that
\begin{align*}
\int \mathcal{A}_\theta g \; \mathrm{d}P_\theta 
 = \int \frac{1}{p_\theta} \nabla \cdot (p_\theta g) \; \mathrm{d}P_\theta
 = \int (\nabla \cdot (p_\theta g))(x) \; \mathrm{d}x = 0,
\end{align*}
as claimed.
\end{proof}

\begin{proposition} \label{lem: Stein kernel}
Let $c : \mathbb{R}^d \times \mathbb{R}^d \rightarrow \mathbb{R}$ be a kernel with $c(\cdot,y), \nabla_y c(\cdot,y), \mathcal{A}_\theta c(\cdot,y) \mathrm{1}_d, \mathcal{A}_\theta \nabla_y c(\cdot,y) \in L^1(P_\theta)$ for each $y \in \mathbb{R}^d$.
Then
\begin{align*}
k(x,y ; \theta) 
& = \nabla_x \cdot \nabla_y c(x,y) - \langle \nabla_x E_\theta(x) , \nabla_y c(x,y) \rangle \\
& \qquad - \langle \nabla_y E_\theta(y) , \nabla_x c(x,y) \rangle + \langle \nabla_x E_\theta(x) , \nabla_y E_\theta(y) \rangle c(x,y) 
\end{align*}
is a $\theta$-dependent kernel with $\mu_k(P_\theta) = 0$ for all $\theta \in \Theta$.
\end{proposition}
\begin{proof}
For the proof that $k(\cdot,\cdot;\theta)$ is indeed a kernel we refer to \citet[][Theorem 2.6]{barp2022targeted}.
For the final part, note that
\begin{align*}
k(x,y ; \theta) = \mathcal{A}_\theta g(x), \qquad g(x) = \nabla_y c(x,y) + c(x,y) \nabla_y \log p_\theta(y)
\end{align*}
where, under our assumptions, $g , \mathcal{A}_\theta g \in L^1(P_\theta)$.
Then we may apply \Cref{lem: Stein op int by parts} to deduce that $\mu_k(P_\theta)(y) = \int k(x,y ; \theta) \mathrm{d}P_\theta(x) = \int \mathcal{A}_\theta g(x) \mathrm{d}P_\theta(x) = 0$ for all $y \in \mathcal{X}$, and thus $\mu_k(P_\theta) = 0$, as claimed.
\end{proof}

\noindent This construction enables minimum kernel discrepancy estimation to be applied to energy-based models, with positive initial result reported in a range of statistical applications \citep{barp2019minimum,Matsubara2022}.
Indeed, $k(\cdot,\cdot;\theta)$ can be evaluated without the normalisation constant $Z(\theta)$, and any kernel for which the conclusion of \Cref{lem: Stein kernel} holds satisfies $k_\theta(\cdot,\cdot) = k(\cdot,\cdot ; \theta)$.
The discrepancy associated to $k(\cdot,\cdot;\theta)$ is called a \emph{kernel Stein discrepancy} \citep{liu2016kernelized,chwialkowski2016kernel,gorham2017measuring}; these kernelise the \emph{score-matching divergence} of \cite{Hyvaerinen2005} and have recently been applied to a variety of other tasks in statistics and machine learning \citep{anastasiou2021stein}.

\smallskip

This completes our short tour of how minimum kernel discrepancy methods can be used.
Next we turn our attention to the asymptotic properties of these estimators, aiming insofar as possible for a general framework.

\section{Asymptotic Theory}
\label{sec: theory}

The aim of this section is to present general asymptotic theory that is applicable to each of the applications we have just discussed.
If the kernel is $\theta$-independent, the minimum kernel discrepancy estimator can be recognised as both an M-estimator and a minimum scoring rule estimator, so established asymptotic theory can be applied \citep{van2000asymptotic,Dawid2016}.
However, this identification does not hold when the kernel is $\theta$-dependent, and bespoke arguments are needed.
To this end, conditions for strong consistency of the minimum kernel discrepancy estimator are established for a $\theta$-independent kernel in \Cref{subsec: strong} and extended to the case of a $\theta$-dependent kernel in \Cref{subsec: dependent case}.
Conditions for asymptotic normality in the case of a possibly $\theta$-dependent kernel are established in \Cref{subsec: asym norm}, and the strength of these conditions is examined in \Cref{subsec: discuss assumptions}.
The presentation is intended to be self-contained; we relate our results to existing literature in \Cref{subsec: related}.

To simplify the presentation, it will be assumed throughout that $P \in \mathcal{P}_k(\mathcal{X})$ and $\{P_\theta\}_{\theta \in \Theta} \subset \mathcal{P}_k(\mathcal{X})$, so that the kernel discrepancy is well-defined.
It will also be assumed that the statistical model is parameterised in such a way that $\theta_\star$ and $\theta_n$ actually exist, in the latter case for $n$ sufficiently large (but in neither case is uniqueness assumed).
Convergence in distribution, convergence in probability, and almost sure convergence will be respectively denoted $\rightarrowd$, $\rightarrowp$ and $\rightarrowas$.

\subsection{Strong Consistency}
\label{subsec: strong}

In this section we consider a topological space $\Theta$; no additional structure on $\Theta$ is required until \Cref{subsec: dependent case}.
Our first task is to find weak sufficient conditions under which a.s. every sequence $(\theta_n)_{n \in \mathbb{N}}$ of minimum kernel discrepancy estimators satisfies
\begin{align}
\{\text{accumulation points of} \; (\theta_n)_{n \in \mathbb{N}} \} \subseteq \argmin_{\theta \in \Theta} D_k(P,P_\theta) . \label{eq: result}
\end{align}
That is, there is a probability 1 set on which all accumulation points of all sequences $\theta_n$ of minimum kernel discrepancy estimators are minimisers $\theta_\star$ of the kernel discrepancy between the distribution $P$, from which the data are sampled, and the statistical model $P_\theta$.
To begin with we consider the case where the kernel $k$ does not depend on the parameter $\theta$ and standard arguments can be used; the case of a $\theta$-dependent kernel is treated in \Cref{subsec: dependent case}.

\begin{lemma}[Strong consistency of the kernel mean element] \label{lem: slln}
Assume that
\begin{assumptionsA}
\item $\mathcal{H}(k)$ is separable. \label{ass: Hk separable}
\item $\int \sqrt{k(x,x)} \; \mathrm{d}P(x) < \infty$.  \label{ass: k is l1}
\end{assumptionsA}
Then $\mu_k(P_n) \rightarrowas \mu_k(P)$.
\end{lemma}
\begin{proof}
The random variables $\xi_i = k(\cdot,x_i)$ are independent, identically distributed, and satisfy $\mathbb{E}[\|k(\cdot,x_i)\|_{\mathcal{H}(k)}] = \int \sqrt{k(x,x)} \; \mathrm{d}P < \infty$ (from \ref{ass: k is l1}), so we may appeal to the strong law of large numbers (\Cref{thm: SLLN}) in the separable (from \ref{ass: Hk separable}) Banach space $\mathcal{B} = \mathcal{H}(k)$, to get
\begin{align*}
\mu_k(P_n) = \frac{1}{n} \sum_{i=1}^n k(\cdot,x_i) \rightarrowas \int k(\cdot,x) \; \mathrm{d}P(x) = \mu_k(P) ,
\end{align*}
which establishes the result.
\end{proof}

\noindent Note that \ref{ass: k is l1} implies the weak integral $\mu_k(P) = \int k(\cdot,x) \; \mathrm{d}P(x)$ is in fact a strong (or \emph{Bochner}) integral.

\begin{lemma}[Uniform convergence in discrepancy] \label{lem: uniform}
Assume \ref{ass: Hk separable}, \ref{ass: k is l1}.
Then
\begin{align}
\sup_{\theta \in \Theta} |D_k(P_n, P_\theta) - D_k(P,P_\theta)| & \rightarrowas 0 . \label{eq: uniform}
\end{align}
\end{lemma}
\begin{proof}
Let $f(\theta) = D_k(P,P_\theta)$ and $f_n(\theta) = D_k(P_n,P_\theta)$.
From the reverse triangle inequality and \Cref{lem: slln}, we have that
\begin{align}
\sup_{\theta \in \Theta} |f_n(\theta) - f(\theta)| & = \sup_{\theta \in \Theta} \left| \|\mu_k(P_n) - \mu_k(P_\theta)\|_{\mathcal{H}(k)} - \|\mu_k(P) - \mu_k(P_\theta)\|_{\mathcal{H}(k)} \right| \nonumber \\
& \leq \|\mu_k(P_n) - \mu_k(P)\|_{\mathcal{H}(k)} \rightarrowas 0 ,
\end{align}
which establishes the result.
\end{proof}

\begin{lemma}[Discrepancy is minimised] \label{thm: basic theorem}
Assume \ref{ass: Hk separable}, \ref{ass: k is l1}.
Then 
$$
D_k(P,P_{\theta_n}) \rightarrowas \min_{\theta \in \Theta} D_k(P,P_\theta) .
$$
\end{lemma}
\begin{proof}
Let $f$ and $f_n$ be as in the proof of \Cref{lem: uniform}.
Let $\epsilon > 0$.
Pick one element $\theta_\star \in \argmin_{\theta \in \Theta} f(\theta)$.
Then a.s. there exists $n_0$ such that for all $n > n_0$ we have that
$$
f(\theta_\star) \stackrel{\text{defn of }\theta_\star}{\leq} f(\theta_n) \stackrel{\text{\Cref{lem: uniform}}}{<} f_n(\theta_n) + \frac{\epsilon}{2} \stackrel{\text{defn of }\theta_n}{\leq} f_n(\theta_\star) + \frac{\epsilon}{2} \stackrel{\text{\Cref{lem: uniform}}}{<} f(\theta_\star) + \epsilon .
$$
Since $\epsilon > 0$ was arbitrary, it follows that $f(\theta_n) \rightarrowas \min_{\theta \in \Theta} f(\theta)$.
\end{proof}

Using an additional assumption that there exists a unique and sufficiently distinct global minimum at the parameter level, we obtain a basic strong consistency result at the parameter level.
The proof of the following result is immediate from \Cref{thm: basic theorem}:

\begin{theorem}[Strong consistency with unique minimum] \label{thm: strong several}
Assume \ref{ass: Hk separable}, \ref{ass: k is l1}.
If $\theta_\star \in \Theta$ is such that for every open neighbourhood $\theta_\star \in N(\theta_\star) \subset \Theta$ we have
\begin{align}
D_k(P,P_{\theta_\star}) < \inf_{\theta \in \Theta \setminus N(\theta_\star)} D_k(P,P_\theta),  \label{eq: van der vaart}
\end{align}
then $\theta_n \rightarrowas \theta_\star$.
\end{theorem}

In practice the uniqueness of a minimum is often difficult or impossible to check, since in general the statistical model will be misspecified.
This motivates a more realistic analysis, which we state in two parts.
The \emph{closure} of a subset $S$ of a topological space $\Theta$ will be denoted $\text{cl}(S)$, and we recall that a subset $C$ of a topological space $\Theta$ is \emph{sequentially compact} if every sequence in $C$ has an accumulation point in $C$.

\begin{theorem}[Strong consistency with several minima, part I] \label{cor: multi 1}
Assume \ref{ass: Hk separable}, \ref{ass: k is l1}, and assume that $\theta \mapsto D_k(P,P_\theta)$ is a continuous function on $\Theta$.
Then any accumulation point of $(\theta_n)_{n \in \mathbb{N}}$ is a.s. an element of $\min_{\theta \in \Theta} D_k(P,P_\theta)$.
\end{theorem}
\begin{proof}
Let $f$ be as in the proof of \Cref{lem: uniform}.
Let $\epsilon > 0$ and let $\theta_\star$ be an accumulation point, meaning that there exists a subsequence with $\theta_{n_m} \rightarrow \theta_\star$ as $m \rightarrow \infty$.
Then a.s. there exists $m_0$ such that for all $m > m_0$ we have
$$
f(\theta_\star)  \stackrel{f\text{ cts at }\theta_\star}{<} f(\theta_{n_m}) + \epsilon .
$$
Taking the limit $m \rightarrow \infty$ on both sides and using \Cref{thm: basic theorem}, we have that a.s.
$$
f(\theta_\star) < \min_{\theta \in \Theta} f(\theta) + \epsilon .
$$
Since $\epsilon > 0$ was arbitrary, it follows that a.s. $\theta_\star \in \min_{\theta \in \Theta} f(\theta)$.
\end{proof}

\begin{theorem}[Strong consistency with several minima, part II] \label{cor: multi 2}
Assume \ref{ass: Hk separable}, \ref{ass: k is l1}, and assume there is an $\epsilon > 0$ for which 
$$
C := \text{\normalfont cl}\left\{ \theta \in \Theta : D_k(P,P_\theta) < \epsilon + \inf_{\vartheta \in \Theta} D_k(P,P_\vartheta) \right\} 
$$
is sequentially compact.
Then the sequence $(\theta_n)_{n \in \mathbb{N}}$ has at least one accumulation point.
\end{theorem}
\begin{proof}
Let $f$ and $f_n$ be as in the proof of \Cref{lem: uniform}.
From \Cref{lem: uniform}, there a.s. exists $n_0$ such that for all $n > n_0$ we have $\sup_{\theta \in \Theta} |f_n(\theta) - f(\theta)| < \epsilon / 2$.
In particular, there a.s. exists $n_0$ such that $(\theta_n)_{n > n_0} \subset C$.
Since the sequence $(\theta_n)_{n \in \mathbb{N}}$ is a.s. eventually contained in $C$, and $C$ is sequentially compact, there a.s. exists an accumulation point.
\end{proof}

\noindent For $\Theta = \mathbb{R}^d$, the topological condition in \Cref{cor: multi 2} is satisfied when $C$ is bounded.

\subsection{Extension to $\theta$-Dependent Kernel}
\label{subsec: dependent case}

To accommodate also the $\theta$-dependent kernels $k(\cdot,\cdot;\theta)$ from \Cref{subsec: minimum KSD}, here we provide an alternative to \Cref{lem: uniform}.

\begin{lemma}[Uniform convergence in discrepancy; $\theta$-dependent kernel] \label{lem: uniform 2}
Assume \ref{ass: Hk separable} for each $\theta \in \Theta$.
Assume $\Theta \subset \mathbb{R}^p$ is open, convex and bounded, and that
\begin{assumptionsA}
\item $\int \sqrt{ k_\theta(x,x) } \; \mathrm{d}P(x) < \infty$, for all $\theta \in \Theta$  \label{ass: ktheta is l1}
\item $\int \sup_{\theta \in \Theta} \left\| \partial_\theta k_\theta(x,x) \right\| \; \mathrm{d}P(x) < \infty$  \label{ass: vstat1}
\item $\iint \sup_{\theta \in \Theta} \left\| \partial_\theta k_\theta(x,y) \right\| \; \mathrm{d}P(x) \mathrm{d}P(y) < \infty$ . \label{ass: vstat2}
\end{assumptionsA}
Then
\begin{align}
\sup_{\theta \in \Theta} |D_\theta(P_n,P_\theta) - D_\theta(P,P_\theta)| & \rightarrowas 0 . \label{eq: uniform}
\end{align}
\end{lemma}
\begin{proof}
The aim is to establish the conditions for the uniform law of large numbers in \Cref{thm: ulln}.
Let $g(\theta) = D_\theta(P,P_\theta)^2$ and $g_n(\theta) = D_\theta(P_n,P_\theta)^2$.
From \eqref{eq: ktheta def} and the fact that $k(\cdot,\cdot;\theta)$ is positive semi-definite,
\begin{align*}
k(x,x;\theta) & \leq k_\theta(x,x) + 2 \int k(x,y; \theta) \mathrm{d}P_\theta(y) .
\end{align*}
Taking the square root and integrating,
\begin{align}
\int \sqrt{k(x,x;\theta)} \mathrm{d}P(x) & \leq \int \sqrt{ k_\theta(x,x) + 2 \int k(x,y; \theta) \mathrm{d}P_\theta(y) } \; \mathrm{d}P(x)  \nonumber \\
& \leq \int \sqrt{ k_\theta(x,x) } \mathrm{d}P(x) + \int \sqrt{ 2 \int k(x,y;\theta) \mathrm{d}P_\theta(y) } \; \mathrm{d}P(x) \label{eq: square root algebra}  \\
& \leq \int \sqrt{ k_\theta(x,x) } \mathrm{d}P(x) + \sqrt{ \int 2 \int k(x,y;\theta) \mathrm{d}P_\theta(y)  \; \mathrm{d}P(x) } \label{eq: use Jensen root} ,
\end{align}
where \eqref{eq: square root algebra} used the fact that $\sqrt{a+b} \leq \sqrt{a} + \sqrt{b}$, and \eqref{eq: use Jensen root} used Jensen's inequality.
The first integral in \eqref{eq: use Jensen root} exists from \ref{ass: ktheta is l1} and all other integrals exist as a consequence of the standing assumption that $P,P_\theta \in \mathcal{P}_{k(\cdot,\cdot;\theta)}(\mathcal{X})$.
Thus $\int \sqrt{k(x,x;\theta)} \; \mathrm{d}P(x) < \infty$.
From this fact and \ref{ass: Hk separable}, we can apply \Cref{lem: slln} to deduce that $g_n(\theta) \rightarrowas g(\theta)$ for each $\theta \in \Theta$.
Thus the first condition in \Cref{thm: ulln} is satisfied.

For the second condition in \Cref{thm: ulln}, we start by noting it is implicit in \ref{ass: vstat1}, \ref{ass: vstat2} that $\partial_\theta k_\theta$ exists in $\Theta$ and thus $g_n$ is differentiable in $\Theta$.
Let $\theta,\vartheta \in \Theta$.
From the mean value theorem, and the fact that $\Theta$ is open and convex, 
\begin{align*}
g_n(\theta) - g_n(\vartheta) & = \langle \theta - \vartheta , \partial_\phi g_n(\phi) \rangle
\end{align*}
for some $\phi \in \{t \theta + (1-t) \vartheta : 0 \leq t \leq 1\} \subset \Theta$.
In particular,
\begin{align*}
\left| g_n(\theta) - g_n(\vartheta) \right| & \leq \| \theta - \vartheta \| \; \sup_{\phi \in \Theta} \| \partial_\phi g_n(\phi) \| 
\leq \| \theta - \vartheta\| \; \frac{1}{n^2} \sum_{i=1}^n \sum_{j=1}^n \underbrace{ \sup_{\phi \in \Theta} \|\partial_\phi k_\phi(x_i,x_j) \|  }_{=: v(x_i,x_j)} .
\end{align*}
Now, since we have a V-statistic, we use \ref{ass: vstat1} and \ref{ass: vstat2} to verify that
\begin{align*}
\iint |v(x,y)| \mathrm{d}P(x) \mathrm{d}P(y) & = \iint \sup_{\theta \in \Theta} \|\partial_\theta k_\theta(x,y)\| \mathrm{d}P(x) \mathrm{d}P(y) < \infty \\
\int |v(x,x)| \mathrm{d}P(x) & = \int \sup_{\theta \in \Theta} \|\partial_\theta k_\theta(x,x)\| \mathrm{d}P(x) < \infty
\end{align*}
and thus from \Cref{thm: pfister C4} we have that 
\begin{align*}
B_n & := \frac{1}{n^2} \sum_{i=1}^n \sum_{j=1}^n \sup_{\theta \in \Theta} \|\partial_\theta k_\theta(x_i,x_j) \|  
\rightarrowas \iint  \sup_{\theta \in \Theta} \|\partial_\theta k_\theta(x,y) \| \mathrm{d}P(x) \mathrm{d}P(y) < \infty .
\end{align*}
Since in addition $\Theta$ is assumed to be bounded, the conditions of the uniform law of large numbers in \Cref{thm: ulln} have now been established.
Thus \Cref{thm: ulln} gives that $\sup_{\theta \in \Theta} |g_n(\theta) - g(\theta)| \rightarrowas 0$, as claimed.
\end{proof}

The arguments of \Cref{thm: basic theorem} and \Cref{thm: strong several,cor: multi 1,cor: multi 2} all go through with \ref{ass: ktheta is l1}, \ref{ass: vstat1} and \ref{ass: vstat2} in place of \ref{ass: k is l1}.

\subsection{Asymptotic Normality}
\label{subsec: asym norm}

The last asymptotic result we present is to show that fluctuations of the minimum kernel discrepancy estimator are asymptotically normal.
For a matrix $M \in \mathbb{R}^{p \times p}$, we write $M \succ 0$ to indicate that $M$ is positive definite; i.e. for all $0 \neq c \in \mathbb{R}^p$ we have $c^\top M c > 0$.

\begin{theorem}[Asymptotic normality] \label{thm: normal}
Suppose $\theta_n \rightarrowas \theta_\star$.
Let there exist an open set $S \subseteq \Theta \subseteq \mathbb{R}^p$ such that $\theta_\star \in S$ and the following hold:
\begin{assumptionsA}
\item $\iint \sup_{\theta \in S} \left\| \partial_\theta k_\theta(x,y) \right\| \mathrm{d}P(x) \mathrm{d}P(y) < \infty$ \label{ass: luib}
\item the functions $\{\theta \mapsto \partial_\theta k_\theta(x,y) : x,y \in \mathcal{X}\}$ are differentiable on $S$ \label{ass: differentiable}
\item the functions $\{ \theta \mapsto \partial_\theta^2 k_\theta(x,y) : x,y \in \mathcal{X} \}$ are uniformly continuous at $\theta_\star$ \label{ass: uniform continuity}
\item $\iint \|\partial_\theta k_\theta(x,y)\|^2 \mathrm{d}P(x) \mathrm{d}P(y) |_{\theta = \theta_\star} < \infty$  \label{ass: d moment}
\item $\iint \| \partial_\theta^2 k_\theta(x,x) \| \mathrm{d}P(x) |_{\theta = \theta_\star} < \infty$ \label{ass: d2 moment diag}
\item $\iint \| \partial_\theta^2 k_\theta(x,y) \| \mathrm{d}P(x) \mathrm{d}P(y) |_{\theta = \theta_\star} < \infty$ \label{ass: d2 moment}  
\item $\Gamma := \frac{1}{2} \left. \iint \partial_\theta^2 k_\theta(x,y) \mathrm{d}P(x) \mathrm{d}P(y) \right|_{\theta = \theta_\star} \succ 0$. \label{ass: Gamma pdef} 
\end{assumptionsA}
Then $\sqrt{n}(\theta_n - \theta_\star) \rightarrowd N(0, \Gamma^{-1} \Sigma \Gamma^{-\top} )$, where $\Sigma := \mathbb{C}_{X \sim P} \left[ \int \partial_\theta k_\theta(X,y) \mathrm{d}P(y) \right] |_{\theta = \theta_\star}$.
\end{theorem}
\begin{proof}
For a function $\theta \mapsto h(\theta)$, let $h'(\theta) = \partial_\theta h(\theta)$, and let $h''(\theta) = \partial_\theta^2 h(\theta)$.
Let $g$ and $g_n$ be as in the proof of \Cref{lem: uniform 2}.
From \ref{ass: differentiable} there is a convex open set $N(\theta_\star) \subset \Theta$ with $\theta_\star \in N(\theta_\star)$ on which the functions $\theta \mapsto \partial_\theta k_\theta(x,y)$ are differentiable, and thus the map $\theta \mapsto \partial_\theta g_n(\theta)$ is differentiable on $N(\theta_\star)$.
Since we are assuming $\theta_n \rightarrowas \theta_\star$, there a.s. exists $n_0 \in \mathbb{N}$ such that for all $n \geq n_0$ we have $\theta_n \in N(\theta_\star)$.
Thus for each $n \geq n_0$ we can apply the mean value theorem on the convex open set $N(\theta_\star)$ to establish that there is a $\tilde{\theta}_n = t_n \theta_n + (1-t_n) \theta_\star$, $t_n \in [0,1]$, for which 
\begin{align*}
0 = g_n'(\theta_n) = g_n'(\theta_\star) + g_n''(\tilde{\theta}_n) \times (\theta_n - \theta_\star)  .
\end{align*}
If $g_n''(\tilde{\theta}_n)$ is non-singular (as will be established below by showing convergence to a matrix $2\Gamma$, assumed to be positive definite due to \ref{ass: Gamma pdef}), then we can re-express this as
\begin{align}
\sqrt{n}(\theta_n - \theta_\star) = - [g_n''(\tilde{\theta}_n)]^{-1} [\sqrt{n} g_n'(\theta_\star)] \label{eq: Taylor}
\end{align}
The stated result follows from Slutsky's lemma if we can show both that $g_n''(\tilde{\theta}_n) \rightarrowp 2 \Gamma$ and $\sqrt{n} g_n'(\theta_\star) \rightarrowd N(0, 4 \Sigma )$.
Both of these will now be established.

\paragraph{Establishing $g_n''(\tilde{\theta}_n) {\normalfont\rightarrowp} 2\Gamma$.}

First we make the trivial algebraic observation that
\begin{align}
g_n''(\tilde{\theta}_n) & = g_n''(\theta_\star) + [ g_n''(\tilde{\theta}_n) - g_n''(\theta_\star) ]  . \label{eq: add subtract}
\end{align}
From \ref{ass: d2 moment diag} and \ref{ass: d2 moment}, for $r,s \in \{1,\dots,p\}$ we have that
\begin{align*}
\left. \int | \partial_{\theta_r} \partial_{\theta_s} k_\theta(x,x) | \mathrm{d}P(x) \right|_{\theta = \theta_\star} , \quad
\left. \iint | \partial_{\theta_r} \partial_{\theta_s} k_\theta(x,y) | \mathrm{d}P(x) \mathrm{d}P(y) \right|_{\theta = \theta_\star} < \infty
\end{align*}
and thus \Cref{thm: pfister C4} on the consistency of V-statistics can be applied, establishing that
\begin{align*}
[g_n''(\theta_\star)]_{r,s} & = \left. \frac{1}{n^2} \sum_{i=1}^n \sum_{j=1}^n \partial_{\theta_r} \partial_{\theta_s} k_\theta(x_i,x_j) \right|_{\theta = \theta_\star} 
\rightarrowas \left. \iint \partial_{\theta_r} \partial_{\theta_s} k_\theta(x,y) \mathrm{d}P(x) \mathrm{d}P(y) \right|_{\theta = \theta_\star} = 2 \Gamma_{r,s} ,
\end{align*}
and in particular $g_n''(\theta_\star) \rightarrowp 2 \Gamma$.
Next consider the second term in \eqref{eq: add subtract}.
For any $\epsilon > 0$, from \ref{ass: uniform continuity} there exists a bounded open set $N_\epsilon(\theta_\star) \subset \Theta$ such that $\theta_\star \in N_\epsilon(\theta_\star)$ and $\|\partial_\theta^2 k_\theta(x,y) - \partial_\theta^2 k_\theta(x,y)|_{\theta = \theta_\star} \| < \epsilon$ for all $\theta \in N_\epsilon(\theta_\star)$ and all $x,y \in \mathcal{X}$.
Since we are assuming $\theta_n \rightarrowas \theta_\star$, this means $\tilde{\theta}_n \rightarrowas \theta_\star$, and thus there a.s. exists $n_\epsilon \in \mathbb{N}$ such that for all $n \geq n_\epsilon$ we have $\tilde{\theta}_n \in N_\epsilon(\theta_\star)$.
Thus, for all $n \geq n_\epsilon$,
\begin{align*}
\left\| g_n''(\tilde{\theta}_n) - g_n''(\theta_\star) \right\| & = \left\| \frac{1}{n^2} \sum_{i=1}^n \sum_{j=1}^n \partial_\theta^2 k_\theta(x_i,x_j) |_{\theta = \tilde{\theta}_n} - \partial_\theta^2 k_\theta(x_i,x_j) |_{\theta = \theta_\star} \right\| \\
& \leq \frac{1}{n^2} \sum_{i=1}^n \sum_{j=1}^n \left\| \partial_\theta^2 k_\theta(x_i,x_j) |_{\theta = \tilde{\theta}_n} - \partial_\theta^2 k_\theta(x_i,x_j) |_{\theta = \theta_\star} \right\| 
 < \epsilon
\end{align*}
which demonstrates that the second term in \eqref{eq: add subtract} a.s. vanishes.

\paragraph{Establishing $\sqrt{n} g_n'(\theta_\star) {\normalfont\rightarrowd} N(0, 4 \Sigma )$.}

Let $0 \neq c \in \mathbb{R}^p$.
From \ref{ass: luib}, the function $h : S \times (\mathcal{X} \times \mathcal{X}) \rightarrow \mathbb{R}^p$ defined as $h(\theta,(x,y)) = k_\theta(x,y)$ has $\partial_\theta h$ locally uniformly integrably bounded in the sense of \Cref{def: LUIB}, and thus the conditions of \Cref{lem: interchange} hold in a neighbourhood of $\theta_\star$, permitting us to interchange $\partial_\theta$ and the integral with respect to $P \times P$.
This, together with the fact that $\theta_\star$ is a minimiser of $g(\theta)$, shows that
\begin{align}
\left. \iint c^\top \partial_\theta k_\theta(x,y) \mathrm{d}P(x) \mathrm{d}P(y) \right|_{\theta = \theta_\star}
& = c^\top \partial_\theta \left. \iint k_\theta(x,y) \mathrm{d}P(x) \mathrm{d}P(y) \right|_{\theta = \theta_\star}
= c^\top g'(\theta_\star) = 0  . \label{eq: gprins is 0}
\end{align}
Further, from \ref{ass: d moment} and Cauchy--Schwarz,
\begin{align*}
\left. \iint (c^\top \partial_\theta k_\theta(x,y) )^2 \mathrm{d}P(x) \mathrm{d}P(y) \right|_{\theta = \theta_\star} \leq \|c\|^2 \left. \iint \| \partial_\theta k_\theta(x,y) \|^2 \mathrm{d}P(x) \mathrm{d}P(y) \right|_{\theta = \theta_\star} < \infty  .
\end{align*}
Thus we can appeal to asymptotic normality of V-statistics, using \Cref{thm: pfister C7}, to obtain that
\begin{align*}
 c^\top \sqrt{n} g_n'(\theta_\star)  
= \sqrt{n} \left( \frac{1}{n^2} \sum_{i=1}^n \sum_{j=1}^n c^\top \partial_\theta k_\theta(x_i,x_j) |_{\theta = \theta_\star} - \underbrace{0}_{\text{from \eqref{eq: gprins is 0}}} \right)
\rightarrowd N(0, 4 \sigma_c^2) 
\end{align*}
where 
\begin{align*}
\sigma_c^2 & = \mathbb{V}_{X \sim P} \left[ \left. \int c^\top \partial_\theta k_\theta(X,y) \mathrm{d}P(y) \right|_{\theta = \theta_\star} \right] 
= c^\top \Sigma c  .
\end{align*}
Since $c \neq 0$ was arbitrary, it follows that $\sqrt{n} g_n'(\theta_\star) = \sqrt{n}[g_n'(\theta_\star) - g'(\theta_\star)] \rightarrowd N(0, 4 \Sigma )$, as required.
\end{proof}

The assumptions of \Cref{thm: normal} will be discussed in \Cref{subsec: discuss assumptions}, but we briefly remark that \ref{ass: vstat2} implies \ref{ass: luib} and that, due to the crude way it appears in the proof, \ref{ass: uniform continuity} can likely be weakened.
The conclusion of \Cref{thm: normal} agrees with our discussion of the genralised method of moments in \Cref{subsec: summary matching}; in this case $\partial_\theta k_\theta(x,y) = - \phi(x) - \phi(y) + 2 \theta$ and $\partial_\theta^2 k_\theta(x,y) = 2 I$, so $\Gamma = I$ and $\Sigma = \mathbb{C}_{X \sim P}[ \phi(X) ]$, and \Cref{thm: normal} gives $\sqrt{n}(\theta_n - \theta_\star) \rightarrowd N(0,\mathbb{C}_{X \sim P}[\phi(X)])$, in agreement.
Finally, we note that \Cref{thm: normal} gives rise to a consistent \emph{sandwich} estimator $\Gamma_n^{-1} \Sigma_n \Gamma_n^{-1}$ for the asymptotic covariance $\Gamma^{-1} \Sigma \Gamma^{-1}$ of the minimum kernel discrepancy estimator, where $\Gamma_n$ and $\Sigma_n$ are obtained by substituting $\theta_n$ \emph{in lieu} of $\theta_\star$ into the definitions of $\Gamma$ and $\Sigma$; see \citet{freedman2006so}.
The resulting $100\gamma\%$ confidence sets for $\theta_\star$ are illustrated in \Cref{fig: gan}.

\subsection{Related Work}
\label{subsec: related}

The theoretical properties of minimum kernel discrepancy estimators have been studied by a number of authors since 2015, but our analysis is arguably the most general to have appeared.
Indeed, as we will discuss, previous work focused on either the case where $k$ does not depend on $\theta$, or on the case where $k$ has a specific form of dependence on $\theta$.

In the case where $k$ does not depend on $\theta$, the original work of \citet{dziugaite2015training} focused on obtaining non-asymptotic concentration inequalities at the level of the kernel discrepancy.
Although not discussed in that work, these guarantees can be extended to guarantees at the level of the parameter under further assumptions similar to those in \Cref{thm: strong several}.
These non-asymptotic results were extended to the regression context, where data $x_1,\dots,x_n$ are non-independent, in \citet{alquier2020universal,cherief2022finite}.
Both a non-asymptotic analysis and an asymptotic analysis were presented in \citet{briol2019statistical}, where strong consistency was established under assumptions that include (1) a unique $\theta_\star$, and (2) for all $n$, there exists $\epsilon_n > 0$ such that $\{\theta \in \Theta : D_k(P_n,P_\theta) < D_k(P,P_{\theta_\star}) + \epsilon_n \}$ is a.s. bounded.
The latter paper is comprehensive and includes also a discussion of the important issues of estimator robustness, the geometry induced by kernel discrepancy, and computational aspects of minimum kernel discrepancy estimation (a point we return to in \Cref{sec: summary}).
A Bayesian interpretation of minimum discrepancy estimation was proposed in \cite{cherief2020mmd}.

The case where $k$ depends on $\theta$ was first studied in \citet{barp2019minimum}.
There the authors focus on asymptotic analysis and a specific form of $k$ derived from Stein's method, the canonical form of which we discussed in \Cref{subsec: minimum KSD}.
Strong consistency of the minimum kernel Stein discrepancy estimator was established under assumptions that include (1) the statistical model is well-specified (i.e. $P = P_{\theta_\star}$), (2) $\theta_\star$ is unique, (3) $\theta \mapsto P_\theta$ is injective, and either (4a) $\Theta$ is compact, or (4b) the maps $\theta \mapsto D_k(P_n,P_\theta)$ are a.s. convex.
Sufficient conditions for asymptotic normality are also provided, and the important issues of estimator robustness, the geometry induced by kernel discrepancy, and computational aspects of minimum kernel Stein discrepancy estimation were discussed.
Closest in spirit to our analysis is \citet{matsubara2022robust}, who endowed the minimum kernel Stein discrepancy estimator with a Bayesian interpetation.
The present analysis aims to be more concise, and also more general in the senses that:
(1) uniqueness of $\theta_\star$ is not assumed, and (2) existence of third derivatives $\partial_\theta^3 k_\theta$ is not assumed.
The aforementioned authors also discussed estimator robustness and computational aspects of minimum kernel Stein discrepancy estimation in detail.

Naturally, many alternatives to kernel discrepancy exist and have been explored in the parameter inference context; to limit scope these have not been discussed.
However, we note that kernel discrepancies form a large class, some of which can be viewed as the limit of entropy-regularised optimal transport \citep{genevay2018learning}, some of which offer control over smoothed Wasserstein distances \citep{nietert2021smooth}, and some of which metrise weak convergence on compact domains \citep{simon2020metrizing}.

\section{Open Research Directions}
\label{sec: summary}

This paper was intended to serve both as a self-contained exposition of minimum kernel discrepancy estimation and as an invitation to contribute to this nascent research field.
Accordingly, in this section we present a non-exhaustive list of open research questions that could be addressed:
\begin{enumerate}
\item \textbf{Computing the kernel mean element.} 
An important issue which we have not yet discussed is how to compute the kernel mean element $\mu_k(P_\theta)$, given that for general choices of $k$ and $P_\theta$ there will not be a closed-form solution to this integral.
In the machine learning community, a Monte Carlo approximation of $P_\theta$ is typically used, often in tandem with stochastic gradient-based optimisation over $\theta \in \Theta$.
However, in some situations one might expect to do better:
In the context of generative models \eqref{eq: gen model}, one has $\mu_k(P_\theta) = \int k(G^\theta(x),\cdot) \mathrm{d}\mathbb{P}(x)$ where $\mathbb{P}$ is a user-specified reference measure, typically the uniform masure on $[0,1]^d$.
This raises the possibility of using quasi Monte Carlo methods to approximate the kernel mean element.
A preliminary investigation in \citet{Niu2021} established the $O(n^{-1 + \epsilon})$ rate, but it remains an open problem to obtain explicit conditions on $G^\theta$ and $k$ that ensure higher-order quasi Monte Carlo convergence rates hold.
Related, although the kernel mean element $\mu_k(P_n)$ is explicit, it is sometimes prohibitive to work with a very large dataset.
In these circumstances, methods have been proposed to select a representative subset of data such that the kernel mean element associated to this subset is close to that of the full datset \citep{cortes2016sparse,teymur2021optimal}.
However, the interaction between computational approximations and the performance of the minimum kernel discrepancy estimator have yet to be studied, beyond the Monte Carlo case in \citet{briol2019statistical}.

\item \textbf{Selecting a kernel.} Beyond the basic requirement that the estimator is consistent, further guidence on the application-specific selection of the kernel $k$ is needed.
The choice of $k$ engenders an efficiency-robustness trade-off, so one must speculate about the extent to which a statistical model is likely to be misspecified.
Further, the efficiency of minimum kernel discrepancy estimation in high-dimensional domains $\mathcal{X}$ is yet to be explored.
It is typical to encounter $P$ whose mass is concentrated around a subset of $\mathcal{X}$, often a low-dimensional linear subspace or a sub-manifold.
In such circumstances it would be desirable to select a kernel $k$ that is appropriately adapted to this low-dimensional set, perhaps as described in \Cref{subsec: kernels in QMC}, to improve the dimension dependence of the estimator.
Insight may be gained by considering the case where $\mathcal{X}$ is infinite-dimensional; kernel discrepancies are well-defined in this context \citep{wynne2022kernel,wynne2022spectral} but minimum kernel discrepancy estimation is yet to be explored.
On the other hand, the case where $\mathcal{X}$ is discrete raises a different set of challenges, since selection of a natural kernel in the discrete context can be difficult.
Both settings are further complicated by the fact that minimum kernel discrepancy estimation is not invariant to how the data are represented, i.e. the choice of $\mathcal{X}$ itself, and both theoretical and practical guidence is needed.

\item \textbf{Insight into machine learning.}  Progress in machine learning is rapid, and in 2021 it was demonstrated that GANs are out-performed by so-called \emph{diffusion models} on a range of image-generation tasks \citep{dhariwal2021diffusion}.
Such applications are characterised by large training datasets, so $P_n$ is approximately equal to $P$, and high-dimensional $\mathcal{X}$, so there is a loose analogy to the setting where quasi Monte Carlo is studied.
One might speculate on whether insights from quasi Monte Carlo might help to explain the limited performance of GANs on such high-dimensional imaging tasks, or even suggest ways in which GANs might be improved.

\item \textbf{Probability in Hilbert spaces.}  The strong consistency of the minimum kernel discrepancy estimator was established, for a $\theta$-independent kernel $k$, as a consequence of the strong law of large numbers in $\mathcal{H}(k)$.
For a $\theta$-independent kernel it may also be possible also to establish asymptotic normality from the Hilbert space version of the central limit theorem.
This approach would arguably be more elegant than the one we presented, avoiding the need to deal with V-statistics.
However, it is not clear to this author whether such a strategy can be applied in the case of a $\theta$-dependent kernel.

\item \textbf{Confidence sets.}  An estimate $\theta_n$ for a parameter of interest $\theta_\star$ is only useful if the precision of $\theta_n$ can be estimated.
\Cref{thm: normal} suggests the asymptotically valid $100 \gamma \%$ confidence set
$$
C_n^\gamma := \{ \theta \in \Theta : \| \Sigma_n^{-1/2} \Gamma_n (\theta - \theta_n) \|^2 \leq F_{\chi_p^2}(\gamma)  \} ,
$$
where for simplicity here we assume that $\Theta = \mathbb{R}^p$ and $\Sigma \succ 0$, and we use $F_\cdot$ notation to denote the cumulative probability distribution, here of a chi-squared distribution with $p$ degrees of freedom, $\chi_p^2$.
However, at finite sample sizes $n$, the coverage of these confidence sets could be far from the notional $100\gamma\%$ due to the variance of the plug-in estimators $\Sigma_n$ and $\Gamma_n$ for $\Sigma$ and $\Gamma$.
This problem is exacerbated when $\theta$ is high-dimensional, since these matrices are of dimension $p \times p$ where $p = \text{dim}(\Theta)$.
A variety of strategies for regularised estimation exist and could be explored in this context \citep{lam2020high}.
\end{enumerate}

\paragraph{Acknowledgements}

The author wishes to thank Alessandro Barp, Fran\c{c}ois-Xavier Briol, Andrew Duncan, Jeremias Knoblauch, Lester Mackey, Takuo Matsubara, and an anonymous reviewer for comments on an earlier version of this manuscript, and financial support from EP/W019590/1 and EP/N510129/1.

\appendix 

\section*{Appendices}

\Cref{app: kernel embed} discusses sufficient conditions for $P \in \mathcal{P}_k(\mathcal{X})$.
\Cref{app: aux} contains auxiliary results from probability and calculus that we used.
\Cref{subsec: discuss assumptions,app: assumptions derive} unpack and verify theoretical assumptions for the exponential family model.

\section{Kernel Mean Embedding} \label{app: kernel embed}

In this section, $P$ is a distribution on a measurable space $\mathcal{X}$, and $k$ is a measurable and symmetric positive definite kernel with reproducing kernel Hilbert space denoted $\mathcal{H}(k)$.
The central issue considered here is when $\mathrm{P}(h) = \int h \; \mathrm{d}P$ is a bounded linear functional on $\mathcal{H}(k)$.
Our presentation here follows \citet[][Appendix C]{barp2022targeted}.

\begin{definition}[Scalarly integrable]
A map $\Phi : \mathcal{X} \rightarrow \mathcal{H}(k)$ is said to be \emph{scalarly} $P$-\emph{integrable} if $\{ x \mapsto \langle h , \Phi(x) \rangle_{\mathcal{H}(k)} : h \in \mathcal{H}(k) \} \subset L^1(P)$.
\end{definition}

\noindent The argument used in the following proof can be traced back to \citet{dunford1937integration}, with our account based on \citet[][Lemma 2.1.1]{schwabik2005topics}.

\begin{proposition} \label{prop: Dunford}
If $\Phi : \mathcal{X} \rightarrow \mathcal{H}(k)$ is scalarly $P$-integrable, then 
\begin{eqnarray*}
\mathrm{P}_\Phi : \mathcal{H}(k) & \rightarrow & \mathbb{R} \\
h & \mapsto & \int \langle h , \Phi(x) \rangle_{\mathcal{H}(k)} \; \mathrm{d}P(x) 
\end{eqnarray*}
is a bounded linear functional.
\end{proposition}
\begin{proof}
First we claim that the graph of the linear map
\begin{eqnarray*}
T : \mathcal{H}(k) & \rightarrow & L^1(P) \\
h & \mapsto & ( x \mapsto \langle h , \Phi(x) \rangle_{\mathcal{H}(k)} ) .
\end{eqnarray*}
is closed.
To see this, let $h_n \rightarrow h$ in $\mathcal{H}(k)$ and suppose that $T(h_n) \rightarrow g$ in $L^1(P)$.
The claim is that $g$ and $T(h)$ are equal in $L^1(P)$.
Since every sequence converging in $L^1(P)$ has an almost surely converging subsequence, there is a subsequence $(h_{n_i})_{i \in \mathbb{N}}$ such that
$$
\langle h_{n_i} , \Phi(x) \rangle_{\mathcal{H}(k)} \rightarrow g(x)
$$
for $P$-almost all $x \in \mathcal{X}$.
Since $\langle h_n , \Phi(x) \rangle_{\mathcal{H}(k)} \rightarrow \langle h , \Phi(x) \rangle_{\mathcal{H}(k)}$ for all $x \in \mathcal{X}$, it follows that $g(x) = \langle h , \Phi(x) \rangle_{\mathcal{H}(k)} = T(h)(x)$ for $P$-almost all $x \in \mathcal{X}$.
Thus $g$ and $T(h)$ are equal in $L^1(P)$ and the graph of $T$ is indeed closed.
The conditions of the closed graph theorem, which states that a linear map $T$ between Banach spaces is bounded if and only if its graph is closed, have now been verified.
Thus $T$ is bounded and we denote
$$
\|T\| := \sup_{\|h\|_{\mathcal{H}(k)} \leq 1} \int |T(h)| \; \mathrm{d}P < \infty .
$$
This allows us to conclude that $\mathrm{P}_\Phi$ is a bounded linear functional, since
\begin{align*}
|\mathrm{P}_\Phi(h)|
= \left| \int \langle h , \Phi(x) \rangle_{\mathcal{H}(k)} \; \mathrm{d}P(x) \right| 
\leq \int | \langle h , \Phi(x) \rangle_{\mathcal{H}(k)} | \; \mathrm{d}P(x) 
= \int |T(h)| \; \mathrm{d}P 
\leq \|T\| \|h\|_{\mathcal{H}(k)} ,
\end{align*}
as required.
\end{proof}

\begin{corollary}[Characterisation of $\mathcal{P}_k(\mathcal{X})$]
$P \in \mathcal{P}_k(\mathcal{X})$ if and only if $\mathcal{H}(k) \subset L^1(P)$.
\end{corollary}
\begin{proof}
Take $\Phi(x) = k(\cdot,x)$ to be the \emph{canonical feature map}, so that $k$ being measurable implies the scalar functions $x \mapsto \langle h , \Phi(x) \rangle = h(x)$ are measurable and $\mathrm{P} = \mathrm{P}_\Phi$.
If $\mathcal{H}(k) \subset L^1(P)$ then $x \mapsto k(\cdot,x)$ is scalarly $P$-integrable, and \Cref{prop: Dunford} shows that $\mathrm{P}$ is a bounded linear functional.
Conversely, if $\mathcal{H}(k) \not\subset L^1(P)$ then it is clear that $\mathrm{P}$ is not bounded and thus $P \notin \mathcal{P}_k(\mathcal{X})$.
\end{proof}

\noindent The question then reduces to when $\mathcal{H}(k) \subset L^1(P)$.

\begin{proposition}
If $\int \sqrt{k(x,x)} \mathrm{d}P(x) < \infty$ then $\mathcal{H}(k) \subset L^1(P)$.
\end{proposition}
\begin{proof}
For $h \in \mathcal{H}(k)$, from the reproducing property and Cauchy--Schwarz,
\begin{align*}
\int |h(x)| \; \mathrm{d}P(x) = \int |\langle h , k(\cdot,x) \rangle | \; \mathrm{d}P(x) \leq \|h\|_{\mathcal{H}(k)} \int \|k(\cdot,x)\|_{\mathcal{H}(k)} \; \mathrm{d}P(x) ,
\end{align*}
where the reproducing property again yields $\|k(\cdot,x)\|_{\mathcal{H}(k)} = \sqrt{k(x,x)}$, as required.
\end{proof}

\noindent One can weaken the above integrability condition under mild assumptions on $k$ and $\mathcal{X}$: 

\begin{proposition} \label{prop: double integral condition}
If $k$ is continuous, $\mathcal{X}$ is separable, and $\iint |k(x,y)| \; \mathrm{d}P(x) \mathrm{d}P(y) < \infty$, then $\mathcal{H}(k) \subset L^1(P)$.
\end{proposition}
\begin{proof}
Since $k$ is measurable and real-valued, $k$ is \emph{strongly measurable} in the sense of \citet[][Section 3.1]{carmeli2006vector}.
Further, since $k$ is strongly measurable and $\iint |k(x,y)| \; \mathrm{d}P(x) \mathrm{d}P(y) < \infty$, then $k$ is \emph{$\infty$-bounded} in the sense of \citet[][Definition 4.1]{carmeli2006vector}; see \citet[][Corollary 4.3]{carmeli2006vector}.
Since $k$ is continuous and $\mathcal{X}$ is separable, it follows that $\mathcal{H}(k)$ is separable \citep[][Corollary 5.2]{carmeli2006vector}.
Since $\mathcal{H}(k)$ is separable, $k$ being $\infty$-bounded is equivalent to $\mathcal{H}(k) \subset L^1(P)$ \citep[][Proposition 4.4]{carmeli2006vector}.
\end{proof}

\section{Auxiliary Results}
\label{app: aux}

The following auxiliary results were used for the proofs in the main text:

\begin{theorem}[Strong law of large numbers in a Banach space] \label{thm: SLLN}
Let $\mathcal{B}$ be a separable Banach space and let $(\zeta_n)_{n \in \mathbb{N}}$ be a sequence of independent and identically distributed random variables taking values in $\mathcal{B}$ such that $\mathbb{E}[\|\zeta_n\|_{\mathcal{B}}] < \infty$ and $\mathbb{E}[\zeta_n] = 0$.
Then
$$
\frac{1}{n} \sum_{i=1}^n \zeta_i \rightarrowas 0 .
$$
\end{theorem}
\begin{proof}
See \citet[][Corollary 7.10]{ledoux1991probability}.
\end{proof}

\begin{theorem}[Strong law of larger numbers for V-statistics] \label{thm: pfister C4}
Let $P$ be a distribution on a measurable space $\mathcal{X}$, and let $(x_n)_{n \in \mathbb{N}}$ be a sequence of independent draws from $P$.
Let $v : \mathcal{X} \times \mathcal{X} \rightarrow \mathbb{R}$ be a measurable symmetric function with $\iint |v(x,y)| \mathrm{d}P(x) \mathrm{d}P(y) < \infty$ and $\int |v(x,x)| \mathrm{d}P(x) < \infty$.
Then 
$$
\frac{1}{n^2} \sum_{i=1}^n \sum_{j=1}^n v(x_i,x_j) \rightarrowas \iint v(x,y) \mathrm{d}P(x) \mathrm{d}P(y) .
$$
\end{theorem}
\begin{proof}
Let
\begin{align*}
V_n := \frac{1}{n^2} \sum_{i=1}^n \sum_{j=1}^n v(x_i,x_j), \qquad U_n := \frac{1}{n(n-1)} \sum_{i=1}^n \sum_{j \neq i} v(x_i,x_j), \qquad R_n := \sum_{i=1}^n v(x_i , x_i) .
\end{align*}
From the usual strong law of large numbers, we have that $n^{-1} R_n \rightarrowas \int v(x,x) \; \mathrm{d}P(x) < \infty$ and thus $n^{-2} R_n \rightarrowas 0$.
From the strong law of large numbers for U-statistics, due to \citet{hoeffding1961strong}, we have that $U_n \rightarrowas \iint v(x,y) \mathrm{d}P(x) \mathrm{d}P(y)$.
Thus
$$
V_n = \frac{n-1}{n} U_n + \frac{1}{n^2} R_n \rightarrowas 1 \times \iint v(x,y) \mathrm{d}P(x) \mathrm{d}P(y) + 0 ,
$$
which completes the argument.
\end{proof}

\begin{theorem}[Asymptotic normality of V-statistics] \label{thm: pfister C7}
Let $P$ be a distribution on a measurable space $\mathcal{X}$, and let $(x_n)_{n \in \mathbb{N}}$ be a sequence of independent draws from $P$.
Let $v : \mathcal{X} \times \mathcal{X} \rightarrow \mathbb{R}$ be a measurable symmetric function with $\iint v(x,y)^2 \mathrm{d}P(x) \mathrm{d}P(y) < \infty$.
Let $\sigma^2 = \mathbb{V}_{X \sim P}[ \int v(X,y) \mathrm{d}P(y) ]$.
Then
$$
\sqrt{n}\left( \frac{1}{n^2} \sum_{i=1}^n \sum_{j=1}^n v(x_i,x_j) - \iint v(x,y) \mathrm{d}P(x) \mathrm{d}P(y) \right) \rightarrowd \mathcal{N}(0, 4 \sigma^2) ,
$$
where we interpret $\mathcal{N}(0,0)$ as a point mass $\delta_0$ at $0 \in \mathcal{X}$.
\end{theorem}
\begin{proof}
This proof uses the same notation as the proof of \Cref{thm: pfister C4}.
From \citet{hoeffding1992class}, since $\iint v(x,y)^2 \mathrm{d}P(x) \mathrm{d}P(y) < \infty$, we have asymptotic normality of the U-statistic
$$
\sqrt{n}\left( U_n - \iint v(x,y) \mathrm{d}P(x) \mathrm{d}P(y) \right) \rightarrowd  \mathcal{N}(0, 4 \sigma^2) 
$$
(see also Section 5.5.1 of \citet{serfling2009approximation}).
From Theorem 1 of \citet{bonner1977note}, since $\iint v(x,y) \mathrm{d}P(x) \mathrm{d}P(y) < \infty$, we have that $\sqrt{n}(V_n - U_n) \rightarrowp 0$ (see also Section 5.7.3 of \citet{serfling2009approximation}).
From Slutsky's theorem, we can combine these facts to arrive at the required result.
\end{proof}

\begin{theorem}[Uniform law of large numbers] \label{thm: ulln}
Let $\Theta \subset \mathbb{R}^p$ be bounded.
Let $f : \Theta \rightarrow \mathbb{R}$ be fixed and consider a sequence $(f_n)_{n \in \mathbb{N}}$ of stochastic functions $f_n : \Theta \rightarrow \mathbb{R}$.
Suppose that 
\begin{itemize}
\item $f_n(\theta) \rightarrowas f(\theta)$ for all $\theta \in \Theta$
\item $|f_n(\theta) - f_n(\vartheta)| \leq B_n \| \theta - \vartheta\|$ for all $\theta, \vartheta \in \Theta$, where $B_n$ does not depend on $\theta, \vartheta$ and a.s. $\limsup B_n < \infty$.
\end{itemize}
Then $\sup_{\theta \in \Theta} |f_n(\theta) - f(\theta)| \rightarrowas 0$.
\end{theorem}
\begin{proof}
This follows as a combination of Theorem 21.8 and Theorem 21.10 from \citet{davidson1994stochastic}.
\end{proof}

\begin{definition} \label{def: LUIB}
Let $P$ be a distribution on a measurable space $\mathcal{X}$.
A function $f : \Theta \times \mathcal{X} \rightarrow \mathbb{R}^p$ is \emph{locally uniformly integrably bounded} with respect to $P$ if, for every $\theta \in \Theta$, there is a non-negative function $b_\theta : \mathcal{X} \rightarrow \mathbb{R}$ that is integrable with respect to $P$, and an open neighbourhood $U_\theta$ of $\theta$, such that $\|f(\vartheta,x)\| \leq b_\theta(x)$ for all $\vartheta \in U_\theta$.
\end{definition}

\begin{lemma}[Differentiate under the Lebesgue integral] \label{lem: interchange}
Let $P$ be a distribution on a measurable space $\mathcal{X}$.
Let $\Theta \subseteq \mathbb{R}^p$ be an open set.
Let $h : \Theta \times \mathcal{X} \rightarrow \mathbb{R}$ be such that $\partial_\theta h$ is locally uniformly integrably bounded with respect to $P$.
Then 
$$
\partial_\theta \int h(\theta,x) \mathrm{d}P(x) = \int \partial_\theta h(\theta,x) \mathrm{d}P(x) .
$$
\end{lemma}
\begin{proof}
See, for example, \citet[][Theorem 24.5, p.193]{aliprantis1998principles}, \citet[][Theorem 16.8, pp.181-182]{billingsley1979probability}.
\end{proof}

\section{Assumptions for the Exponential Family Model}
\label{subsec: discuss assumptions}

To better understand \ref{ass: Hk separable}-\ref{ass: Gamma pdef}, we consider the consequences of these assumptions in the setting of canonical exponential family statistical model 
$$
p_\theta(x) = \exp(\langle \theta , t(x) \rangle - a(\theta) +b(x))
$$
where $x \in \mathcal{X} = \mathbb{R}^d$, whose parameter $\theta \in \Theta \subset \mathbb{R}^p$ is estimated using either the generalised method of moments (\Cref{subsub: assum gmm}) or using a kernel Stein discrepancy (\Cref{subsub: assum ksd}).

\subsection{Estimation via Generalised Method of Moments}
\label{subsub: assum gmm}

Consider the generalised method of moments from \Cref{subsec: summary matching}, which uses a finite rank kernel of the form $k(x,y) = \langle \phi(x),\phi(y) \rangle$.
To simplify the following discussion we focus on the well-specified case where there is a unique $\theta_\star \in \Theta = \mathbb{R}^p$ with $P = P_{\theta_\star}$ and seek conditions under which the generalised moment matching estimator is consistent, with fluctuations that are asymptotically normal.
In this context there are essentially two important conditions that must be satisfied:
\begin{assumptionsB}
\item $\phi$ is uniformly continuous and bounded \label{ass: k uniform conts}
\item $M_{i,j} := \left. \int \phi_j(x) [t_i(x) - \partial_{\theta_i} a(\theta)] \mathrm{d}P(x) \right|_{\theta = \theta_\star}$ has full row rank.  \label{ass: M nonsingular}
\end{assumptionsB}
The first part, \ref{ass: k uniform conts}, is satisfied by (for example) the \emph{Fourier features} popularised in \citet{Rahimi2007}, while \ref{ass: M nonsingular} ensures that the set of features is rich enough to identify the correct model.
In addition to these two main conditions, we require sufficient regularity that, in a bounded open neighbourhood $S$ of $\theta_\star$:
\begin{assumptionsB}
\item $\theta \mapsto \int \phi \; \mathrm{d}P_\theta$ is continuous at each $\theta \in S$ \label{assm: cts features}
\item $\theta \mapsto \int t \; \mathrm{d}P_\theta$ and $\int t t^\top \; \mathrm{d}P_\theta$ are continuous at $\theta_\star$ \label{assum: t and tt cts}
\item $\sup_{\theta \in S} \int \|t\|^2 \; \mathrm{d}P_\theta < \infty$ \label{ass t2 integrable}
\item for each $\theta \in S$, there exists $\gamma_\theta > 0$ such that $\int (1 + \|t\|^2) \exp(\gamma_\theta \|t\|) \mathrm{d}P_\theta < \infty$  \label{ass: exponential moment}
\item $\partial_\theta^2 a(\theta)$ exists for each $\theta \in S$ and is continuous at $\theta_\star$, \label{ass: atheta2 exists}
\end{assumptionsB}
which are all reasonably mild.
Collectively \ref{ass: k uniform conts}-\ref{ass: atheta2 exists} imply \ref{ass: Hk separable}-\ref{ass: k is l1} and \ref{ass: luib}-\ref{ass: Gamma pdef}, so that the generalised moment matching estimator is both consistent and asymptotically normal.
Full derivations are reserved for \Cref{sec: verify}.

\subsection{Estimation via Kernel Stein Discrepancy}
\label{subsub: assum ksd}

Consider the minimum kernel Stein discrepancy estimator from \Cref{subsec: minimum KSD}.
Here we suppose that $\Theta$ is an open, convex and bounded subset of $\mathbb{R}^p$ but, in contrast to the last section, we do not assume here the statistical model is well-specified.
For $f : \mathbb{R}^d \rightarrow \mathbb{R}^d$, the convention $[\nabla f]_{i,j} = \partial_{x_i} f_j(x)$ will be used.
In this application, there are three main conditions to be satisfied:
\begin{assumptionsC}
\item $c(x,y), \nabla_x c(x,y), \nabla_y c(x,y)$ are bounded over $x,y \in \mathbb{R}^d$ \label{assumC: bd c}
\item $\nabla E, \nabla t, \nabla b$ and $x \mapsto \nabla_x \cdot \nabla_y c(x,y)$, $y \in \mathbb{R}^d$, are in $L^1(P_\theta)$ for each $\theta \in \Theta$  \label{assumC:  integrability}
\item $\iint c(x,y) [\nabla t(x)]^\top [\nabla t(y)]  \; \mathrm{d}P(x) \mathrm{d}P(y) \succ 0$  \label{assum: ksd pdef}
\end{assumptionsC}
The requirements \ref{assumC: bd c}-\ref{assumC:  integrability} ensure the integrability conditions of \Cref{lem: Stein kernel} are satisfied, while \ref{assum: ksd pdef} is analogous to \ref{ass: M nonsingular}, ensuring that the optimal statistical model can be identified.
In addition, the following regularity is required:
\begin{assumptionsC}
\item $(x,y) \mapsto c(x,y), \nabla_x c(x,y), \nabla_y c(x,y), \nabla_x \cdot \nabla_y c(x,y)$ are continuous \label{assum: nnk cts}
\item $\nabla t$ and $\nabla b$ are continuous \label{assum: nb, nt}
\item $\int \|\nabla t\|^4 \; \mathrm{d}P$, $\int \|\nabla t\|^3 \|\nabla b\| \; \mathrm{d}P$, $\int \|\nabla t\|^2 \|\nabla b\|^2 \; \mathrm{d}P$, $\int \|\nabla b\| \; \mathrm{d}P < \infty$, \label{assum: t2 bd}
\end{assumptionsC}
which are again reasonably mild.
Collectively \ref{assumC: bd c}-\ref{assum: t2 bd} imply \ref{ass: Hk separable} and \ref{ass: ktheta is l1}-\ref{ass: Gamma pdef}, so that the minimum kernel Stein discrepancy estimator is both consistent and asymptotically normal.
Full derivations are reserved for \Cref{sec: verify 2}

\section{Proofs for the Exponential Family Model} 
\label{app: assumptions derive}

This appendix contains full details on how our theoretical assumptions were verified for for the canonical exponential family statistical model in \Cref{subsec: discuss assumptions}.

\subsection{Estimation via Generalised Method of Moments}
\label{sec: verify}

In this section we verify that collectively \ref{ass: k uniform conts}-\ref{ass: atheta2 exists} imply \ref{ass: Hk separable}-\ref{ass: k is l1} and \ref{ass: luib}-\ref{ass: Gamma pdef}, so that the generalised moment matching estimator is both consistent and asymptotically normal.
The matrix norm $\|M\|^2 = \sum M_{i,j}^2$ will be used.
From \ref{ass: k uniform conts}, we have $|k(x,y)| \leq k_{\max}$ for all $x,y \in \mathbb{R}^p$ where $k_{\max}$ is a positive constant.

\bigskip
\noindent
\textbf{Verify \ref{ass: Hk separable}:}
The space $\mathcal{H}(k)$ is separable whenever $k$ is continuous and $\mathcal{X}$ is separable (which is the case for $\mathcal{X} = \mathbb{R}^d$), and thus \ref{ass: k uniform conts} implies \ref{ass: Hk separable}.

\bigskip
\noindent
\textbf{Verify \ref{ass: k is l1}:}
\ref{ass: k uniform conts} immediately implies \ref{ass: k is l1}.

\bigskip
\noindent
\textbf{Verify \ref{ass: luib} and \ref{ass: d moment}:}
At each $\theta \in S$, we claim the functions $x \mapsto k(x,y) \partial_\theta p_\theta(x)$ and $x \mapsto \int k(x,y) \partial_\theta p_\theta(x)  \mathrm{d}P_\vartheta(y)$ are locally uniformly integrably bounded with respect to the Lebesgue measure on $\mathbb{R}^p$.
Indeed, for the first of these functions we have a bound
\begin{align*}
\int \sup_{\vartheta \in S_{\gamma,\theta}} \| k(x,y) \partial_\vartheta p_\vartheta(x) \| \; \mathrm{d}x & \leq k_{\max} \int \sup_{\vartheta \in S_{\gamma,\theta}} \| t(x) - \partial_\vartheta a(\vartheta) \| p_\vartheta(x) \; \mathrm{d}x \\
& = k_{\max} \int \sup_{\vartheta \in S_{\gamma,\theta}} \| t(x) - \partial_\vartheta a(\vartheta) \| \frac{ p_\vartheta(x) }{ p_\theta(x) }  \; \mathrm{d}P_\theta(x) ,
\end{align*}
where $S_{\gamma,\theta} = S \cap \{\vartheta \in \mathbb{R}^p : \|\vartheta - \theta\| < \gamma_\theta\}$.
An identical bound holds also for the second function considered.
Now,
\begin{align*}
\sup_{\vartheta \in S_{\gamma,\theta}} \| t(x) - \partial_\vartheta a(\vartheta) \| & \leq \|t(x)\| + C_1 \\
\frac{ p_\vartheta(x) }{ p_\theta(x) } & = \exp\left( \langle \vartheta - \theta , t(x) \rangle - [a(\vartheta) - a(\theta)] \right) 
\leq \exp\left( C_2 \|t(x)\| + C_3 \right)
\end{align*}
where $C_1 = \sup_{\vartheta \in S_{\gamma,\theta}} \|\partial_\vartheta a(\vartheta)\|$, $C_2 = \sup_{\vartheta \in S_{\gamma,\theta}} \|\vartheta - \theta\|$ and $C_3 = \sup_{\vartheta \in S_{\gamma,\theta}} |a(\vartheta) - a(\theta)|$.
From \ref{ass: atheta2 exists}, we have $C_1, C_3 < \infty$, and $C_2 < \infty$ holds since $S$ is bounded.
Thus
\begin{align*}
\int \sup_{\vartheta \in S_{\gamma,\theta}} \| k(x,y) \partial_\vartheta p_\vartheta(x) \| \; \mathrm{d}x & \leq k_{\max} \int \left[ \| t(x) \| + C_1 \right] \exp\left( C_2 \|t(x)\| + C_3 \right)  \; \mathrm{d}P_\theta(x) < \infty ,
\end{align*}
where the final integral is finite due to \ref{ass: exponential moment}.
Thus we have established the conditions of \Cref{lem: interchange} and we may interchange $\partial_\theta$ with integration with respect to $\mathrm{d}x$ and $\mathrm{d}y$:
\begin{align*}
\partial_\theta k_\theta(x,y) & = - \int k(x,y) \partial_\theta p_\theta(x) \mathrm{d}x - \int k(x,y) \partial_\theta p_\theta(y) \mathrm{d}y + \iint k(x,y) \partial_\theta[p_\theta(x) p_\theta(y)] \mathrm{d}x \mathrm{d}y \\
& = - \int k(x,y) [t(x)-\partial_\theta a(\theta)] \mathrm{d}P_\theta(x) - \int k(x,y) [t(y) - \partial_\theta a(\theta)] \mathrm{d}P_\theta(y) \\
& \qquad + \iint k(x,y) [t(x) + t(y) - 2 \partial_\theta a(\theta)] \mathrm{d}P_\theta(x) \mathrm{d}P_\theta(y)  
\end{align*}
so, using the triangle inequality and Jensen's inequality,
\begin{align*}
\frac{\|\partial_\theta k_\theta(x,y)\|}{k_{\max}} & \leq 4 \int \|t(x)\| \mathrm{d}P_\theta(x) + 4 \|\partial_\theta a(\theta)\|  .
\end{align*}
From \ref{ass t2 integrable} we have that $\sup_{\theta \in S} \int \|t\| \mathrm{d}P_\theta < \infty$ and, since $S$ is bounded, \ref{ass: atheta2 exists} implies that $\sup_{\theta \in S} \| \partial_\theta a(\theta) \| < \infty$.
Thus $\sup_{x,y \in \mathcal{X}} \sup_{\theta \in S} \|\partial_\theta k_\theta(x,y)\| < \infty$ and, in particular, \ref{ass: luib} and \ref{ass: d moment} hold.

\bigskip
\noindent
\textbf{Verify \ref{ass: differentiable}:}
An argument analogous to the previous argument shows that \ref{ass: exponential moment} is sufficient to allow a second interchange of $\partial_\theta$ with integration with respect to $\mathrm{d}x$ and $\mathrm{d}y$:
\begin{align}
\partial_\theta^2 k_\theta(x,y) 
& = [\partial_\theta^2 a(\theta)] \iint k(x,y) \mathrm{d}P_\theta(x) - \int k(x,y) [t(x) - \partial_\theta a(\theta)] [t(x) - \partial_\theta a(\theta)]^\top \mathrm{d}P_\theta(x) \nonumber \\
& \qquad + [\partial_\theta^2 a(\theta)] \iint k(x,y) \mathrm{d}P_\theta(y) - \int k(x,y) [t(y) - \partial_\theta a(\theta)] [t(y) - \partial_\theta a(\theta)]^\top \mathrm{d}P_\theta(y) \nonumber \\
& \qquad - 2 [\partial_\theta^2 a(\theta)] \iint k(x,y) \mathrm{d}P_\theta(x) \mathrm{d}P_\theta(y)  \label{eq: dtheta2 big equal} \\
& \qquad + \iint k(x,y) [t(x) + t(y) - 2 \partial_\theta a(\theta)] [t(x) + t(y) - 2 \partial_\theta a(\theta)]^\top \mathrm{d}P_\theta(x) \mathrm{d}P_\theta(y) \nonumber
\end{align}
where the quantities exist on $\theta \in S$ due to \ref{ass t2 integrable} and \ref{ass: atheta2 exists}, showing that \ref{ass: differentiable} is satisfied.

\bigskip
\noindent
\textbf{Verify \ref{ass: uniform continuity}:}
Since \eqref{eq: dtheta2 big equal} takes the form
\begin{align*}
\partial_\theta^2 k_\theta(x,y) & = \int k(x,u) f(u,\theta) \; \mathrm{d}u + \int k(y,u) f(u,\theta) \; \mathrm{d}u - \iint k(u,v) g(u,v,\theta) \; \mathrm{d}u \mathrm{d}v \\
f(u,\theta) & := \left\{ [\partial_\theta^2 a(\theta)] - [t(u) - \partial_\theta a(\theta)] [t(u) - \partial_\theta a(\theta)]^\top \right\} p_\theta(u) \\
g(u,v,\theta) & = \left\{ 2 [\partial_\theta^2 a(\theta)] - [t(u) + t(v) - 2 \partial_\theta a(\theta)] [t(u) + t(v) - 2 \partial_\theta a(\theta)]^\top  \right\} p_\theta(u) p_\theta(v) ,
\end{align*}
we have that
\begin{align*}
\sup_{x,y \in \mathcal{X}} \frac{ \| \partial_\theta^2 k_\theta(x,y) - [ \partial_\theta^2 k_\theta(x,y) |_{\theta=\theta_\star} ] \|}{k_{\max}} & \leq 2 \int | f(u,\theta) - f(u,\theta_\star) | \; \mathrm{d}u \\
& \qquad + \iint |g(u,v,\theta) - g(u,v,\theta_\star)| \; \mathrm{d}u \mathrm{d}v
\end{align*}
Since $\partial_\theta^2 a(\theta)$ is continuous at $\theta_\star$ from \ref{ass: atheta2 exists} and $\int t \; \mathrm{d}P_\theta$ and $\int t t^\top \; \mathrm{d}P_\theta$ are continuous at $\theta_\star$ from \ref{assum: t and tt cts}, it follows that both integrals on the right hand side vanish as $\theta \rightarrow \theta_\star$, meaning \ref{ass: uniform continuity} is established.

\bigskip
\noindent
\textbf{Verify \ref{ass: d2 moment diag} and \ref{ass: d2 moment}:}
Application of the triangle inequality to \eqref{eq: dtheta2 big equal}, and collecting together terms, yields
\begin{align*}
\frac{\|\partial_\theta^2 k_\theta(x,y)\|}{k_{\max}} & \leq 4 \int \|t(x)\|^2 \mathrm{d}P_\theta(x) + 12 \|\partial_\theta a(\theta) \| \int \|t(x)\| \mathrm{d}P_\theta(x) \\
& \qquad + 4 \|\partial_\theta^2 a(\theta) \| + 6 \|\partial_\theta a(\theta) \|^2  + 2 \left( \int \|t(x)\| \mathrm{d}P_\theta(x) \right)^2  .
\end{align*}
Thus \ref{ass t2 integrable} and \ref{ass: atheta2 exists} imply $\sup_{x,y \in \mathcal{X}} \left[ \left. \|\partial_\theta^2 k_\theta(x,y)\| \right|_{\theta = \theta_\star} \right] < \infty$, so \ref{ass: d2 moment diag} and \ref{ass: d2 moment} hold.

\bigskip
\noindent
\textbf{Verify \ref{ass: Gamma pdef}:}
Integrating both arguments of \eqref{eq: dtheta2 big equal} with respect to $P = P_{\theta_\star}$, we obtain
\begin{align*}
\Gamma & = \frac{1}{2} \left. \iint \partial_{\theta}^2 k_{\theta}(x,y)  \mathrm{d}P(x) \mathrm{d}P(y) \right|_{\theta = \theta_\star} \\
& = \iint k(x,y) [t(x) - \partial_\theta a(\theta)] [t(y) - \partial_\theta a(\theta)]^\top \; \mathrm{d}P(x) \mathrm{d}P(y) 
= M M^\top \succ 0 ,
\end{align*}
where the final equality uses the definition of $M$ and the fact that $k(x,y) = \langle \phi(x) , \phi(y) \rangle$, and the positive definiteness follows from \ref{ass: M nonsingular}.

\bigskip
\noindent
\textbf{Verify $\theta_n {\normalfont \rightarrowas} \theta_\star$:}
Through the above, we have proven that if $\theta_n \rightarrowas \theta_\star$, then the fluctuations of $\theta_n$ are asymptotically normal.
However, we have not yet commented on whether $\theta_n \rightarrowas \theta_\star$.
Since we are in the well-specified regime, we have from \eqref{eq: dtheta2 big equal}
\begin{align*}
\left. \partial_\theta^2 D_k(P,P_\theta) \right|_{\theta = \theta^*} = \left. \iint \partial_{\theta}^2 k_{\theta}(x,y)  \mathrm{d}P(x) \mathrm{d}P(y) \right|_{\theta = \theta_\star} = 2 \Gamma \succ 0 ,
\end{align*}
showing that $\theta_\star$ is a local minimum of the kernel discrepancy.
Since we assumed $\theta_\star$ is unique, and from \ref{assm: cts features} the map $\theta \mapsto D_k(P,P_\theta)$ is continuous, the conditions of \Cref{thm: strong several} are satisfied, establishing that $\theta_n \rightarrowas \theta_\star$.

\subsection{Estimation via Kernel Stein Discrepancy}
\label{sec: verify 2}

In this section we verify that collectively \ref{assumC: bd c}-\ref{assum: t2 bd} imply \ref{ass: Hk separable} and \ref{ass: ktheta is l1}-\ref{ass: Gamma pdef}, so that the minimum kernel Stein discrepancy estimator is both consistent and asymptotically normal.
From \ref{assumC: bd c}, we have $|c(x,y)|, \|\nabla_x c(x,y)\|, \|\nabla_y c(x,y)\| \leq c_{\max}$ for all $x,y \in \mathbb{R}^d$ where $c_{\max}$ is a positive constant.
From \ref{assumC:  integrability} the integrability conditions of \Cref{lem: Stein kernel} are satisfied, so that $k(\cdot,\cdot;\theta) = k_\theta(\cdot,\cdot)$ for each $\theta \in \Theta$.

\bigskip
\noindent
\textbf{Verify \ref{ass: Hk separable}:}
From \ref{assum: nnk cts} and \ref{assum: nb, nt}, $(x,y) \mapsto k(x,y;\theta) = k_\theta(x,y)$ is continuous, since
\begin{align*}
k_\theta(x,y) & = \nabla_x \cdot \nabla_y c(x,y) + \langle \nabla_x c(x,y) , \nabla t(y) \theta + \nabla b(y) \rangle \\
& \qquad + \langle \nabla_y c(x,y) , \nabla t(x) \theta + \nabla b(x) \rangle + c(x,y) \langle \nabla t(x) \theta + \nabla b(x) , \nabla t(y) \theta + \nabla b(y) \rangle .
\end{align*}
The space $\mathcal{H}(k)$ is separable whenever $k$ is continuous and $\mathcal{X}$ is separable (which is the case for $\mathcal{X} = \mathbb{R}^d$), and thus $\mathcal{H}(k)$ is separable and we have \ref{ass: Hk separable}.

\bigskip
\noindent
\textbf{Verify \ref{ass: differentiable}:}
First we compute
\begin{align}
\partial_{\theta_i} k_\theta(x,y) & = \langle \nabla_x c(x,y) , [\nabla t(y)]_{\cdot,i} \rangle + \langle \nabla_y c(x,y) , [\nabla t(x)]_{\cdot,i} \rangle \label{eq: ktheta express for ksd} \\
& \qquad + c(x,y) \langle [\nabla t(x)]_{\cdot,i} , \nabla t(y) \theta + \nabla b(y) \rangle + c(x,y) \langle \nabla t(x) \theta + \nabla b(x) , [\nabla t(y)]_{\cdot, i} \rangle  , \nonumber
\end{align}
which is a linear function of $\theta$, so that we trivially have \ref{ass: differentiable}.

\bigskip
\noindent
\textbf{Verify \ref{ass: vstat1}, \ref{ass: vstat2}, \ref{ass: luib}, \ref{ass: d moment}:}
The triangle and Cauchy--Schwarz inequalities applied to \eqref{eq: ktheta express for ksd} yield
\begin{align*}
\frac{\|\partial_{\theta_i} k_\theta(x,y)\|}{c_{\max}} & \leq \|\nabla t(x) \| + \| \nabla t(y) \| + \|\nabla t(x) \| \left( \| \nabla t(y) \| \sup_{\theta \in \Theta} \|\theta\| + \|\nabla b(y)\| \right) \\
& \qquad + \|\nabla t(y) \| \left( \| \nabla t(x) \| \sup_{\theta \in \Theta} \|\theta\| + \|\nabla b(x)\| \right) ,
\end{align*}
where $\sup_{\theta \in \Theta} \|\theta\| < \infty$ since we assumed that $\Theta$ was bounded.
Thus \ref{ass: vstat1} holds from \ref{assum: t2 bd}, and both \ref{ass: vstat2} and \ref{ass: luib} hold from \ref{assum: t2 bd}.
Squaring, we see also that \ref{ass: d moment} holds from \ref{assum: t2 bd}. 

\bigskip
\noindent
\textbf{Verify \ref{ass: uniform continuity}, \ref{ass: d2 moment diag}, \ref{ass: d2 moment}:}
Differentiating \eqref{eq: ktheta express for ksd} again,
\begin{align*}
\partial_{\theta_i} \partial_{\theta_j} k_\theta(x,y) & = 2 c(x,y) \langle [\nabla t(x)]_{\cdot,i} , [\nabla t(y)]_{\cdot,j} \rangle
\end{align*}
so \ref{ass: uniform continuity} holds trivially, since these functions are constant in $\theta$.
Further,
\begin{align*}
\frac{\|\partial_{\theta_i} \partial_{\theta_j} k_\theta(x,y) \|}{c_{\max}} & \leq 2 \| \nabla t(x) \| \| \nabla t(y) \|
\end{align*}
so \ref{ass: d2 moment diag} and \ref{ass: d2 moment} hold from \ref{assum: t2 bd}.

\bigskip
\noindent
\textbf{Verify \ref{ass: Gamma pdef}:}
Finally,
\begin{align*}
\Gamma = \frac{1}{2} \left. \iint \partial_\theta^2 k_\theta(x,y) \; \mathrm{d}P(x) \mathrm{d}P(y) \right|_{\theta = \theta_\star} 
= \iint c(x,y) [\nabla t(x)]^\top [\nabla t(y)]  \; \mathrm{d}P(x) \mathrm{d}P(y) \succ 0
\end{align*}
from \ref{assum: ksd pdef}, establishing \ref{ass: Gamma pdef}.

\bigskip
\noindent
\textbf{Verify $\theta_n {\normalfont \rightarrowas} \theta_\star$:}
Since we have verified \ref{ass: ktheta is l1}, \ref{ass: vstat1} and \ref{ass: vstat2}, we have from \Cref{lem: uniform 2} and an essentially identical argument to \Cref{thm: basic theorem} that $D_k(P,P_{\theta_n}) \rightarrowas D_k(P,P_{\theta_\star})$.
Since the squared kernel Stein discrepancy
$$
\theta \mapsto D_k(P,P_\theta)^2 = \iint k_\theta(x,y) \; \mathrm{d}P(x) \mathrm{d}P(y)
$$ 
is a quadratic with minimum $\theta_\star$ and Hessian at $\theta_\star$ equal to $2 \Gamma \succ 0$, it is necessarily the case that $\theta_n \rightarrowas \theta_\star$ in the $n \rightarrow \infty$ limit.


\begin{thebibliography}{79}
\providecommand{\natexlab}[1]{#1}
\providecommand{\url}[1]{\texttt{#1}}
\expandafter\ifx\csname urlstyle\endcsname\relax
  \providecommand{\doi}[1]{doi: #1}\else
  \providecommand{\doi}{doi: \begingroup \urlstyle{rm}\Url}\fi

\bibitem[Akaike(1973)]{Akaike1973}
H.~Akaike.
\newblock Information theory and an extension of the likelihood principle.
\newblock In \emph{Proceedings of the Second International Symposium of
  Information Theory}, 1973.

\bibitem[Aliprantis and Burkinshaw(1998)]{aliprantis1998principles}
C.~D. Aliprantis and O.~Burkinshaw.
\newblock \emph{Principles of Real Analysis}.
\newblock Academic Press, 1998.

\bibitem[Alquier and Gerber(2023)]{alquier2020universal}
P.~Alquier and M.~Gerber.
\newblock Universal robust regression via maximum mean discrepancy.
\newblock \emph{Biometrika}, 2023.
\newblock To appear.

\bibitem[Anastasiou et~al.(2023)Anastasiou, Barp, Briol, Ebner, Gaunt,
  Ghaderinezhad, Gorham, Gretton, Ley, Liu, Mackey, Oates, Reinert, and
  Swan]{anastasiou2021stein}
A.~Anastasiou, A.~Barp, F.-X. Briol, B.~Ebner, R.~E. Gaunt, F.~Ghaderinezhad,
  J.~Gorham, A.~Gretton, C.~Ley, Q.~Liu, L.~Mackey, C.~J. Oates, G.~Reinert,
  and Y.~Swan.
\newblock {S}tein's method meets statistics: {A} review of some recent
  developments.
\newblock \emph{Statistical Science}, 38\penalty0 (1):\penalty0 120--139, 2023.

\bibitem[Arjovsky et~al.(2017)Arjovsky, Chintala, and
  Bottou]{arjovsky2017wasserstein}
M.~Arjovsky, S.~Chintala, and L.~Bottou.
\newblock Wasserstein generative adversarial networks.
\newblock In \emph{Proceedings of the 34th International Conference on Machine
  Learning}, 2017.

\bibitem[Barp et~al.(2019)Barp, Briol, Duncan, Girolami, and
  Mackey]{barp2019minimum}
A.~Barp, F.-X. Briol, A.~Duncan, M.~Girolami, and L.~Mackey.
\newblock Minimum {S}tein discrepancy estimators.
\newblock In \emph{Proceedings of the 33rd Conference on Neural Information
  Processing Systems}, 2019.

\bibitem[Barp et~al.(2022)Barp, Simon-Gabriel, Girolami, and
  Mackey]{barp2022targeted}
A.~Barp, C.-J. Simon-Gabriel, M.~Girolami, and L.~Mackey.
\newblock Targeted separation and convergence with kernel discrepancies.
\newblock \emph{arXiv:2209.12835}, 2022.

\bibitem[Basu et~al.(2011)Basu, Shioya, and Park]{Basu2011}
A.~Basu, H.~Shioya, and C.~Park.
\newblock \emph{Statistical Inference: {T}he Minimum Distance Approach}.
\newblock CRC Press, 2011.

\bibitem[Beaumont(2019)]{beaumont2019approximate}
M.~A. Beaumont.
\newblock Approximate bayesian computation.
\newblock \emph{Annual Review of Statistics and its Application}, 6:\penalty0
  379--403, 2019.

\bibitem[Billingsley(1979)]{billingsley1979probability}
P.~Billingsley.
\newblock \emph{Probability and Measure}.
\newblock John Wiley and Sons, 1979.

\bibitem[Bi{\'n}kowski et~al.(2018)Bi{\'n}kowski, Sutherland, Arbel, and
  Gretton]{binkowski2018demystifying}
M.~Bi{\'n}kowski, D.~J. Sutherland, M.~Arbel, and A.~Gretton.
\newblock Demystifying {MMD} {GAN}s.
\newblock In \emph{Proceedings of the 6th International Conference on Learning
  Representations}, 2018.

\bibitem[Bonner and Kirschner(1977)]{bonner1977note}
N.~Bonner and H.-P. Kirschner.
\newblock Note on conditions for weak convergence of von {M}ises'
  differentiable statistical functions.
\newblock \emph{The Annals of Statistics}, 5\penalty0 (2):\penalty0 405--407,
  1977.

\bibitem[Briol et~al.(2019)Briol, Barp, Duncan, and
  Girolami]{briol2019statistical}
F.-X. Briol, A.~Barp, A.~B. Duncan, and M.~Girolami.
\newblock Statistical inference for generative models with maximum mean
  discrepancy.
\newblock \emph{arXiv:1906.05944}, 2019.

\bibitem[Carmeli et~al.(2006)Carmeli, De~Vito, and Toigo]{carmeli2006vector}
C.~Carmeli, E.~De~Vito, and A.~Toigo.
\newblock Vector valued reproducing kernel {H}ilbert spaces of integrable
  functions and {M}ercer theorem.
\newblock \emph{Analysis and Applications}, 4\penalty0 (04):\penalty0 377--408,
  2006.

\bibitem[Ch{\'e}rief-Abdellatif and Alquier(2020)]{cherief2020mmd}
B.-E. Ch{\'e}rief-Abdellatif and P.~Alquier.
\newblock {MMD}-{B}ayes: {R}obust {B}ayesian estimation via maximum mean
  discrepancy.
\newblock In \emph{Symposium on Advances in Approximate Bayesian Inference},
  pages 1--21. PMLR, 2020.

\bibitem[Ch{\'e}rief-Abdellatif and Alquier(2022)]{cherief2022finite}
B.-E. Ch{\'e}rief-Abdellatif and P.~Alquier.
\newblock Finite sample properties of parametric {MMD} estimation: {R}obustness
  to misspecification and dependence.
\newblock \emph{Bernoulli}, 28\penalty0 (1):\penalty0 181--213, 2022.

\bibitem[Chwialkowski et~al.(2016)Chwialkowski, Strathmann, and
  Gretton]{chwialkowski2016kernel}
K.~Chwialkowski, H.~Strathmann, and A.~Gretton.
\newblock A kernel test of goodness of fit.
\newblock In \emph{Proceedings of the 33rd International Conference on Machine
  Learning}, 2016.

\bibitem[Cortes and Scott(2016)]{cortes2016sparse}
E.~C. Cortes and C.~Scott.
\newblock Sparse approximation of a kernel mean.
\newblock \emph{IEEE Transactions on Signal Processing}, 65\penalty0
  (5):\penalty0 1310--1323, 2016.

\bibitem[Davidson(1994)]{davidson1994stochastic}
J.~Davidson.
\newblock \emph{Stochastic Limit Theory: An Introduction for Econometricians}.
\newblock OUP Oxford, 1994.

\bibitem[Dawid(2007)]{dawid2007geometry}
A.~P. Dawid.
\newblock The geometry of proper scoring rules.
\newblock \emph{Annals of the Institute of Statistical Mathematics},
  59\penalty0 (1):\penalty0 77--93, 2007.

\bibitem[Dawid et~al.(2016)Dawid, Musio, and Ventura]{Dawid2016}
A.~P. Dawid, M.~Musio, and L.~Ventura.
\newblock Minimum scoring rule inference.
\newblock \emph{Scandinavian Journal of Statistics}, 43\penalty0 (1):\penalty0
  123--138, 2016.

\bibitem[Dellaporta et~al.(2022)Dellaporta, Knoblauch, Damoulas, and
  Briol]{dellaporta2022robust}
C.~Dellaporta, J.~Knoblauch, T.~Damoulas, and F.-X. Briol.
\newblock Robust {B}ayesian inference for simulator-based models via the {MMD}
  posterior bootstrap.
\newblock In \emph{Proceedings of the 25th International Conference on
  Artificial Intelligence and Statistics}, 2022.

\bibitem[Dhariwal and Nichol(2021)]{dhariwal2021diffusion}
P.~Dhariwal and A.~Nichol.
\newblock Diffusion models beat {GAN}s on image synthesis.
\newblock In \emph{Proceedings of the 35th Conference on Neural Information
  Processing Systems}, 2021.

\bibitem[Dick and Pillichshammer(2010)]{dick2010digital}
J.~Dick and F.~Pillichshammer.
\newblock \emph{Digital Nets and Sequences: {D}iscrepancy Theory and
  Quasi-{M}onte {C}arlo integration}.
\newblock Cambridge University Press, 2010.

\bibitem[Dick et~al.(2013)Dick, Kuo, and Sloan]{dick2013high}
J.~Dick, F.~Y. Kuo, and I.~H. Sloan.
\newblock High-dimensional integration: {T}he quasi-{M}onte {C}arlo way.
\newblock \emph{Acta Numerica}, 22:\penalty0 133--288, 2013.

\bibitem[Donoho and Liu(1988)]{Donoho1988}
D.~L. Donoho and R.~C. Liu.
\newblock The ``automatic'' robustness of minimum distance functionals.
\newblock \emph{The Annals of Statistics}, 16\penalty0 (2):\penalty0 552--586,
  1988.

\bibitem[Dunford(1937)]{dunford1937integration}
N.~Dunford.
\newblock Integration of vector-valued functions.
\newblock \emph{Bulletin of the American Mathematical Society}, page~43, 1937.

\bibitem[Dziugaite et~al.(2015)Dziugaite, Roy, and
  Ghahramani]{dziugaite2015training}
G.~K. Dziugaite, D.~M. Roy, and Z.~Ghahramani.
\newblock Training generative neural networks via maximum mean discrepancy
  optimization.
\newblock In \emph{Proceedings of the 31st Conference on Uncertainty in
  Artificial Intelligence}, 2015.

\bibitem[Frazier and Drovandi(2021)]{frazier2021robust}
D.~T. Frazier and C.~Drovandi.
\newblock Robust approximate {B}ayesian inference with synthetic likelihood.
\newblock \emph{Journal of Computational and Graphical Statistics}, 30\penalty0
  (4):\penalty0 958--976, 2021.

\bibitem[Freedman(2006)]{freedman2006so}
D.~A. Freedman.
\newblock On the so-called ``{H}uber sandwich estimator'' and ``robust standard
  errors''.
\newblock \emph{The American Statistician}, 60\penalty0 (4):\penalty0 299--302,
  2006.

\bibitem[Genevay et~al.(2018)Genevay, Peyr{\'e}, and
  Cuturi]{genevay2018learning}
A.~Genevay, G.~Peyr{\'e}, and M.~Cuturi.
\newblock Learning generative models with sinkhorn divergences.
\newblock In \emph{Proceedings of the 21st International Conference on
  Artificial Intelligence and Statistics}, 2018.

\bibitem[Gneiting and Raftery(2007)]{gneiting2007strictly}
T.~Gneiting and A.~E. Raftery.
\newblock Strictly proper scoring rules, prediction, and estimation.
\newblock \emph{Journal of the American statistical Association}, 102\penalty0
  (477):\penalty0 359--378, 2007.

\bibitem[Goodfellow et~al.(2020)Goodfellow, Pouget-Abadie, Mirza, Xu,
  Warde-Farley, Ozair, Courville, and Bengio]{goodfellow2020generative}
I.~Goodfellow, J.~Pouget-Abadie, M.~Mirza, B.~Xu, D.~Warde-Farley, S.~Ozair,
  A.~Courville, and Y.~Bengio.
\newblock Generative adversarial networks.
\newblock \emph{Communications of the ACM}, 63\penalty0 (11):\penalty0
  139--144, 2020.

\bibitem[Gorham and Mackey(2017)]{gorham2017measuring}
J.~Gorham and L.~Mackey.
\newblock Measuring sample quality with kernels.
\newblock In \emph{Proceedings of the 34th International Conference on Machine
  Learning}, 2017.

\bibitem[Gretton et~al.(2012)Gretton, Borgwardt, Rasch, Sch{\"o}lkopf, and
  Smola]{gretton2012kernel}
A.~Gretton, K.~M. Borgwardt, M.~J. Rasch, B.~Sch{\"o}lkopf, and A.~Smola.
\newblock A kernel two-sample test.
\newblock \emph{The Journal of Machine Learning Research}, 13\penalty0
  (1):\penalty0 723--773, 2012.

\bibitem[Hansen(1982)]{hansen1982large}
L.~P. Hansen.
\newblock Large sample properties of generalized method of moments estimators.
\newblock \emph{Econometrica}, pages 1029--1054, 1982.

\bibitem[Hickernell(1998)]{hickernell1998generalized}
F.~Hickernell.
\newblock A generalized discrepancy and quadrature error bound.
\newblock \emph{Mathematics of Computation}, 67\penalty0 (221):\penalty0
  299--322, 1998.

\bibitem[Hlawka(1961)]{hlawka1961funktionen}
E.~Hlawka.
\newblock Funktionen von beschr{\"a}nkter variatiou in der theorie der
  gleichverteilung.
\newblock \emph{Annali di Matematica Pura ed Applicata}, 54\penalty0
  (1):\penalty0 325--333, 1961.

\bibitem[Hoeffding(1948)]{hoeffding1992class}
W.~Hoeffding.
\newblock A class of statistics with asymptotically normal distribution.
\newblock \emph{The Annals of Mathematical Statistics}, 19\penalty0
  (3):\penalty0 293--325, 1948.

\bibitem[Hoeffding(1961)]{hoeffding1961strong}
W.~Hoeffding.
\newblock The strong law of large numbers for ${U}$-statistics.
\newblock Technical report, North Carolina State University. Dept. of
  Statistics, 1961.

\bibitem[Huber(1964)]{huber1964robust}
P.~J. Huber.
\newblock Robust estimation of a location parameter.
\newblock \emph{The Annals of Mathematical Statistics}, pages 73--101, 1964.

\bibitem[Hyv{\"a}rinen and Dayan(2005)]{Hyvaerinen2005}
A.~Hyv{\"a}rinen and P.~Dayan.
\newblock Estimation of non-normalized statistical models by score matching.
\newblock \emph{Journal of Machine Learning Research}, 6\penalty0 (4), 2005.

\bibitem[Key et~al.(2021)Key, Fernandez, Gretton, and Briol]{key2021composite}
O.~Key, T.~Fernandez, A.~Gretton, and F.-X. Briol.
\newblock Composite goodness-of-fit tests with kernels.
\newblock \emph{arXiv:2111.10275}, 2021.

\bibitem[Kuo(2003)]{kuo2003component}
F.~Y. Kuo.
\newblock Component-by-component constructions achieve the optimal rate of
  convergence for multivariate integration in weighted {K}orobov and {S}obolev
  spaces.
\newblock \emph{Journal of Complexity}, 19\penalty0 (3):\penalty0 301--320,
  2003.

\bibitem[Lam(2020)]{lam2020high}
C.~Lam.
\newblock High-dimensional covariance matrix estimation.
\newblock \emph{Wiley Interdisciplinary Reviews: Computational Statistics},
  12\penalty0 (2):\penalty0 e1485, 2020.

\bibitem[LeCun et~al.(2007)LeCun, Chopra, Hadsell, Ranzato, and
  Huang]{lecun2006tutorial}
Y.~LeCun, S.~Chopra, R.~Hadsell, M.~Ranzato, and F.~Huang.
\newblock A tutorial on energy-based learning.
\newblock In B.~Sch{\"o}lkopf, A.~J. Smola, B.~Taskar, and S.~Vishwanathan,
  editors, \emph{Predicting Structured Data}. 2007.

\bibitem[Ledoux and Talagrand(1991)]{ledoux1991probability}
M.~Ledoux and M.~Talagrand.
\newblock \emph{Probability in Banach Spaces: Isoperimetry and Processes}.
\newblock Springer Science \& Business Media, 1991.

\bibitem[Li et~al.(2017)Li, Chang, Cheng, Yang, and P{\'o}czos]{li2017mmd}
C.-L. Li, W.-C. Chang, Y.~Cheng, Y.~Yang, and B.~P{\'o}czos.
\newblock {MMD GAN}: {T}owards deeper understanding of moment matching network.
\newblock In \emph{Proceedings of the 31st Conference on Neural Information
  Processing Systems}, 2017.

\bibitem[Li et~al.(2015)Li, Swersky, and Zemel]{Li2015}
Y.~Li, K.~Swersky, and R.~Zemel.
\newblock Generative moment matching networks.
\newblock In \emph{Proceedings of the 32nd International Conference on Machine
  Learning}, 2015.

\bibitem[Liu et~al.(2016)Liu, Lee, and Jordan]{liu2016kernelized}
Q.~Liu, J.~Lee, and M.~Jordan.
\newblock A kernelized {S}tein discrepancy for goodness-of-fit tests.
\newblock In \emph{Proceedings of the 33rd International Conference on Machine
  Learning}, 2016.

\bibitem[Lyne et~al.(2015)Lyne, Girolami, Atchad{\'e}, Strathmann, and
  Simpson]{lyne2015russian}
A.-M. Lyne, M.~Girolami, Y.~Atchad{\'e}, H.~Strathmann, and D.~Simpson.
\newblock On {R}ussian roulette estimates for {B}ayesian inference with
  doubly-intractable likelihoods.
\newblock \emph{Statistical Science}, 30\penalty0 (4):\penalty0 443--467, 2015.

\bibitem[Matsubara et~al.(2022{\natexlab{a}})Matsubara, Knoblauch, Briol, and
  Oates]{Matsubara2022}
T.~Matsubara, J.~Knoblauch, F.-X. Briol, and C.~J. Oates.
\newblock Robust generalised bayesian inference for intractable likelihoods.
\newblock \emph{Journal of the Royal Statistical Society, Series B},
  84\penalty0 (3):\penalty0 997--1022, 2022{\natexlab{a}}.

\bibitem[Matsubara et~al.(2022{\natexlab{b}})Matsubara, Knoblauch, Briol, and
  Oates]{matsubara2022robust}
T.~Matsubara, J.~Knoblauch, F.-X. Briol, and C.~J. Oates.
\newblock Robust generalised {B}ayesian inference for intractable likelihoods.
\newblock \emph{Journal of the Royal Statistical Society: Series B},
  84\penalty0 (3):\penalty0 997--1022, 2022{\natexlab{b}}.

\bibitem[Mitrovic et~al.(2016)Mitrovic, Sejdinovic, and Teh]{mitrovic2016dr}
J.~Mitrovic, D.~Sejdinovic, and Y.-W. Teh.
\newblock {DR-ABC}: {A}pproximate {B}ayesian computation with kernel-based
  distribution regression.
\newblock In \emph{Proceedings of the 33rd International Conference on Machine
  Learning}, 2016.

\bibitem[Mroueh and Sercu(2017)]{mroueh2017fisher}
Y.~Mroueh and T.~Sercu.
\newblock Fisher {GAN}.
\newblock In \emph{Proceedings of the 31st Conference on Neural Information
  Processing Systems}, 2017.

\bibitem[Mroueh et~al.(2017)Mroueh, Sercu, and Goel]{mroueh2017mcgan}
Y.~Mroueh, T.~Sercu, and V.~Goel.
\newblock {McGAN}: {M}ean and covariance feature matching {GAN}.
\newblock In \emph{Proceedings of the 34th International Conference on Machine
  Learning}, 2017.

\bibitem[Mroueh et~al.(2018)Mroueh, Li, Sercu, Raj, and
  Cheng]{mroueh2018sobolev}
Y.~Mroueh, C.-L. Li, T.~Sercu, A.~Raj, and Y.~Cheng.
\newblock Sobolev {GAN}.
\newblock In \emph{Proceedings of the 6th International Conference on Learning
  Representations}, 2018.

\bibitem[Muandet et~al.(2017)Muandet, Fukumizu, Sriperumbudur, and
  Sch{\"o}lkopf]{muandet2016kernel}
K.~Muandet, K.~Fukumizu, B.~Sriperumbudur, and B.~Sch{\"o}lkopf.
\newblock Kernel mean embedding of distributions: {A} review and beyond.
\newblock \emph{Foundations and Trends{\textregistered} in Machine Learning},
  10\penalty0 (1-2):\penalty0 1--141, 2017.

\bibitem[M{\"u}ller(1997)]{muller1997integral}
A.~M{\"u}ller.
\newblock Integral probability metrics and their generating classes of
  functions.
\newblock \emph{Advances in Applied Probability}, 29\penalty0 (2):\penalty0
  429--443, 1997.

\bibitem[Nietert et~al.(2021)Nietert, Goldfeld, and Kato]{nietert2021smooth}
S.~Nietert, Z.~Goldfeld, and K.~Kato.
\newblock Smooth $ p $-wasserstein distance: Structure, empirical
  approximation, and statistical applications.
\newblock In \emph{Proceedings of the 38th International Conference on Machine
  Learning}, 2021.

\bibitem[Niu et~al.(2023)Niu, Meier, and Briol]{Niu2021}
Z.~Niu, J.~Meier, and F.-X. Briol.
\newblock Discrepancy-based inference for intractable generative models using
  quasi-{M}onte {C}arlo.
\newblock \emph{Electronic Journal of Statistics}, 17\penalty0 (1):\penalty0
  1411--1456, 2023.

\bibitem[Oates et~al.(2017)Oates, Girolami, and Chopin]{oates2017control}
C.~J. Oates, M.~Girolami, and N.~Chopin.
\newblock Control functionals for {M}onte {C}arlo integration.
\newblock \emph{Journal of the Royal Statistical Society, Series B},
  79:\penalty0 695--718, 2017.

\bibitem[Pardo(2018)]{Pardo2018}
L.~Pardo.
\newblock \emph{Statistical Inference Based on Divergence Measures}.
\newblock Chapman and Hall/CRC, 2018.

\bibitem[Park et~al.(2016)Park, Jitkrittum, and Sejdinovic]{park2016k2}
M.~Park, W.~Jitkrittum, and D.~Sejdinovic.
\newblock {K2-ABC}: {A}pproximate {B}ayesian computation with kernel
  embeddings.
\newblock In \emph{Proceedings of the 18th International Conference on
  Artificial Intelligence and Statistics}, 2016.

\bibitem[Rahimi and Recht(2007)]{Rahimi2007}
A.~Rahimi and B.~Recht.
\newblock Random features for large-scale kernel machines.
\newblock In \emph{Proceedings of the 21st Conference on Neural Information
  Processing Systems}, 2007.

\bibitem[Schwabik and Ye(2005)]{schwabik2005topics}
S.~Schwabik and G.~Ye.
\newblock \emph{Topics in Banach Space Integration}.
\newblock World Scientific, 2005.

\bibitem[Serfling(2009)]{serfling2009approximation}
R.~J. Serfling.
\newblock \emph{Approximation Theorems of Mathematical Statistics}.
\newblock John Wiley \& Sons, 2009.

\bibitem[Simon-Gabriel et~al.(2023)Simon-Gabriel, Barp, and
  Mackey]{simon2020metrizing}
C.-J. Simon-Gabriel, A.~Barp, and L.~Mackey.
\newblock Metrizing weak convergence with maximum mean discrepancies.
\newblock \emph{Journal of Machine Learning Research}, 24:\penalty0 1--20,
  2023.

\bibitem[Sloan and Kachoyan(1987)]{sloan1987lattice}
I.~H. Sloan and P.~J. Kachoyan.
\newblock Lattice methods for multiple integration: {T}heory, error analysis
  and examples.
\newblock \emph{SIAM Journal on Numerical Analysis}, 24\penalty0 (1):\penalty0
  116--128, 1987.

\bibitem[Sloan and Wo{\'z}niakowski(1998)]{sloan1998quasi}
I.~H. Sloan and H.~Wo{\'z}niakowski.
\newblock When are quasi-{M}onte {C}arlo algorithms efficient for high
  dimensional integrals?
\newblock \emph{Journal of Complexity}, 14\penalty0 (1):\penalty0 1--33, 1998.

\bibitem[Song et~al.(2008)Song, Zhang, Smola, Gretton, and
  Sch{\"o}lkopf]{Song2008}
L.~Song, X.~Zhang, A.~Smola, A.~Gretton, and B.~Sch{\"o}lkopf.
\newblock Tailoring density estimation via reproducing kernel moment matching.
\newblock In \emph{Proceedings of the 25th International Conference on Machine
  Learning}, 2008.

\bibitem[Song and Kingma(2021)]{song2021train}
Y.~Song and D.~P. Kingma.
\newblock How to train your energy-based models.
\newblock \emph{arXiv:2101.03288}, 2021.

\bibitem[Steinwart and Christmann(2008)]{steinwart2008support}
I.~Steinwart and A.~Christmann.
\newblock \emph{Support Vector Machines}.
\newblock Springer Science \& Business Media, 2008.

\bibitem[Sutherland et~al.(2017)Sutherland, Tung, Strathmann, De, Ramdas,
  Smola, and Gretton]{Sutherland2017}
D.~J. Sutherland, H.-Y. Tung, H.~Strathmann, S.~De, A.~Ramdas, A.~J. Smola, and
  A.~Gretton.
\newblock Generative models and model criticism via optimized maximum mean
  discrepancy.
\newblock In \emph{Proceedings of the 5th International Conference on Learning
  Representations}, 2017.

\bibitem[Teymur et~al.(2021)Teymur, Gorham, Riabiz, and
  Oates]{teymur2021optimal}
O.~Teymur, J.~Gorham, M.~Riabiz, and C.~J. Oates.
\newblock Optimal quantisation of probability measures using maximum mean
  discrepancy.
\newblock In \emph{Proceedings of the 24th International Conference on
  Artificial Intelligence and Statistics}, 2021.

\bibitem[Theis et~al.(2016)Theis, van~den Oord, and Bethge]{theis2016note}
L.~Theis, A.~van~den Oord, and M.~Bethge.
\newblock A note on the evaluation of generative models.
\newblock In \emph{Proceedings of the 4th International Conference on Learning
  Representations}, 2016.

\bibitem[Van~der Vaart(2000)]{van2000asymptotic}
A.~W. Van~der Vaart.
\newblock \emph{Asymptotic Statistics}.
\newblock Cambridge University Press, 2000.

\bibitem[Wynne and Duncan(2022)]{wynne2022kernel}
G.~Wynne and A.~B. Duncan.
\newblock A kernel two-sample test for functional data.
\newblock \emph{Journal of Machine Learning Research}, 23\penalty0
  (73):\penalty0 1--51, 2022.

\bibitem[Wynne et~al.(2022)Wynne, Kasprzak, and Duncan]{wynne2022spectral}
G.~Wynne, M.~Kasprzak, and A.~B. Duncan.
\newblock A spectral representation of kernel {S}tein discrepancy with
  application to goodness-of-fit tests for measures on infinite dimensional
  {H}ilbert spaces.
\newblock \emph{arXiv:2206.04552}, 2022.

\end{thebibliography}

\end{document}